\documentclass{article}
\usepackage{geometry,url, graphicx} 
\usepackage{amsmath,amssymb,amsthm,enumerate,color}
\usepackage{epstopdf}
\usepackage{newcent}
\def\d{{\rm d}}
\def\e{{\rm e}}
\def\i{{\rm i}}
\def\ds{\displaystyle}
\def\PI{\mbox{\rm P$_{\rm I}$}}
\def\PII{\mbox{\rm P$_{\rm II}$}}
\def\PIII{\mbox{\rm P$_{\rm III}$}}

\def\PIIIp{\mbox{\rm P$_{\rm III'}$}}

\def\PV{\mbox{\rm P$_{\rm V}$}}
\def\PVI{\mbox{\rm P$_{\rm VI}$}}

\def\sPII{\mbox{\rm S$_{\rm II}$}}
\def\sPIII{\mbox{\rm S$_{\rm III}$}}

\def\HII{\mbox{$\mathcal{H}_{\rm II}$}}
\def\HIII{\mbox{$\mathcal{H}_{\rm III}$}}

\def\a{\alpha}
\def\b{\beta}
\def\ga{\gamma}
\def\de{\delta}
\def\ep{\varepsilon}
\def\ph{\varphi}
\def\k{\kappa}
\def\w{\omega}

\def\p{Painlev\'e}

\def\odes{ordinary differential equations}
\def\peq{\p\ equation}
\def\peqs{\p\ equations}
\def\bk{B\"ack\-lund}
\def\bt{B\"ack\-lund transformation}
\def\bts{B\"ack\-lund transformations}
\def\Z{\mathbb{Z}}
\def\N{\mathbb{N}}

\def\hide#1{}
\def\ifrac#1#2{{#1/#2}}

\newcommand{\HyperpFq}[2]{{}_{#1}F_{#2}}

\newcommand{\BesselI}[1]{I_{#1}}

\newcommand{\Integer}{\mathbb{Z}}
\def\Z{\Integer}
\def\luk{Lukashevich}
\newcommand{\deriv}[3][]{\frac{\d^{#1}{#2}}{{\d{#3}}^{#1}}}
\newcommand{\pderiv}[3][]{\frac{\partial^{#1}{#2}}{{\partial{#3}}^{#1}}}

\def\wz{\deriv{w}{z}}
\def\wzz{\deriv[2]{w}{z}}
\def\O{\mathcal{O}}
\def\W{\mathcal{W}}
\newcommand{\Wr}{\operatorname{Wr}}

\newtheorem{theorem}{Theorem}[section]
\newtheorem{defn}[theorem]{Definition}
\newtheorem{corollary}[theorem]{Corollary}

\newtheorem{lemma}[theorem]{Lemma}
\newtheorem{remark}[theorem]{Remark}
\newtheorem{remarks}[theorem]{Remarks}
\numberwithin{figure}{section}
\numberwithin{equation}{section}
\numberwithin{table}{section}

\def\U{{U}}
\newcommand{\comment}[1]{}
\def\beq{\begin{equation}}
\def\eeq{\end{equation}}

\def\figs#1{}

\begin{document}

\title{Classical Solutions of the Degenerate\\ Fifth \p\ Equation}

\author{Peter A. Clarkson\\[2.5pt]
School of Mathematics, Statistics and Actuarial Science,\\ University of Kent, Canterbury, CT2 7FS, UK\\
Email: {\tt P.A.Clarkson@kent.ac.uk}}

\maketitle

\begin{abstract}In this paper classical solutions of the degenerate fifth Painlev\'e equation are classified, which include hierarchies of algebraic solutions and solutions expressible in terms of Bessel functions. Solutions of the degenerate fifth Painlev\'e equation are known to be expressible in terms of solutions of the third Painlev\'e equation. The classification and description of the classical solutions of the degenerate fifth Painlev\'e equation is done using the Hamiltonian associated with third Painlev\'e equation. Two applications of these classical solutions are discussed, deriving exact solutions of the complex sine-Gordon equation and of the coefficients in the three-term recurrence relation associated with generalised Charlier polynomials.
\end{abstract}

\section{Introduction}

In this paper we are concerned with solutions of the equation
\beq\wzz= \left(\frac{1}{2w} + \frac{1}{w-1}\right)\left(\wz\right)^{\!2} -
\frac{1}{z}\wz + \frac{(w-1)^2(\a w^2+\b)}{z^2w} + \frac{\ga w}{z},
\label{eq:degPV}\eeq
with $\a$, $\b$ and $\ga$ constants.
Equation \eqref{eq:degPV} is the special case of the fifth \p\ equation (\PV) 
\beq\wzz= \left(\frac{1}{2w} + \frac{1}{w-1}\right)\left(\wz\right)^{\!2} -
\frac{1}{z}\wz + \frac{(w-1)^2(\a w^2+\b)}{z^2w} + \frac{\ga w}{z} + \frac{\de w(w+1)}{w-1}.
\label{eq:PVgen}\eeq
with $\a$, $\b$, $\ga$ and $\de$ constants, when $\de=0$ and is known as the \textit{degenerate fifth \peq}\ (deg-\PV), cf.~\cite{refOO06}.

The six \peqs\ (\PI--\PVI), were discovered by \p, Gambier and their colleagues whilst studying second order ordinary differential equations of the form
\beq \label{gen-ode}
\deriv[2]{w}{z}=F\left(z,w,\deriv{w}{z}\right), 
\eeq 
where $F$ is rational in $\d w/\d z$ and $w$ and analytic in $z$. The {\p\ functions} can be thought of as nonlinear analogues of the classical special functions. 
\comment{Indeed Iwasaki, Kimura, Shimomura and Yoshida \cite{refIKSY} characterize the six \peqs\ as ``the most important nonlinear ordinary differential equations" and state that ``many specialists believe that during the twenty-first century the \p\ functions will become new members of the community of special functions". Subsequently this has happened as the \p\ equations are a chapter in the NIST \textit{Digital Library of Mathematical Functions}\ \cite[\S32]{refDLMF}.}%
The general solutions of the \peqs\ are transcendental in the sense that they cannot be expressed in terms of known elementary functions and so require the introduction of a new transcendental function to describe their solution. However, it is well known that \PII--\PVI\ possess rational solutions, algebraic solutions and solutions expressed in terms of the classical special functions --- Airy, Bessel, parabolic cylinder, Kummer and hypergeometric functions, respectively --- for special values of the parameters, see, e.g.\ \cite{
refPACreview,refGLS}
and the references therein. These hierarchies are usually generated from ``seed solutions'' using the associated \bk\ transformations and frequently can be expressed in the form of determinants. These solutions of the \p\ equations are often called ``classical solutions", cf.~\cite{refUmemura98,refUmemura01}.

It is well known that solutions of deg-\PV\ \eqref{eq:degPV} are related to solutions of the third \p\ equation 
\beq\label{eq:PIIIgen}
\deriv[2]{q}{x} = \frac1q \left(\deriv{q}{x}\right)^{\!2}- \frac{1}{x} \deriv{q}{x}+\frac{A q^2+B}{x} + C q^3 + \frac{D}{q},
\eeq
with $A$, $B$, $C$ and $D$ constants, a result originally due to Gromak \cite{refGromak75}; see also \cite[\S34]{refGLS}. The relationship between solutions of deg-\PV\ and the third \peq\ is given in Lemma \ref{lem21} below.
The objective of this paper is to give a classification and description of the classical solutions of deg-\PV\ \eqref{eq:degPV} using the associated Hamiltonian formalism, rather than through solutions of the third \peq\   \eqref{eq:PIIIgen}.

In \S\ref{sec2}, the relationship between deg-\PV\ \eqref{eq:degPV} and the third \peq\ \eqref{eq:PIIIgen} is discussed using the associated Hamitonian. 
In \S\ref{sec:p3sol}, classical solutions of the third \peq\ \eqref{eq:PIIIgen} are reviewed, the rational solutions in \S\ref{sec:p3sol1} and the Bessel function solutions in \S\ref{sec:p3sol2}. 
In \S\ref{sec:bts}, \bts\ of deg-\PV\ \eqref{eq:degPV} are given, which can be used to derive a hierarchy of solutions from a ``seed solution". 
In \S\ref {sec:dpvsol}, classical solutions of deg-\PV\ \eqref{eq:degPV} are classified, the algebraic solutions in \S\ref{sec:dpvsol1} and the Bessel function solutions in \S\ref{sec:dpvsol2}. 
In \S\ref{sec:app}, two applications of classical solutions of deg-\PV\ \eqref{eq:degPV} are given to derive exact solutions of the complex sine-Gordon equation, which is equivalent to the Pohlmeyer-Lund-Regge model,  and to derive explicit representations of the coefficients in the three-term recurrence relation satisfied by generalised Charlier polynomials, which are discrete orthogonal {polynomials}. 

\section{\label{sec2}The relationship between deg-\PV\ and \PIII}
In the generic case when $CD\not=0$ in the third \p\ equation 
\eqref{eq:PIIIgen},  we set $C=1$ and $D=-1$, without loss of generality (by rescaling the variables if necessary), and so consider the equation 
\beq \label{eq:PIII}
\deriv[2]{q}{x} = \frac1q \left(\deriv{q}{x}\right)^{\!2}- \frac{1}{x} \deriv{q}{x}
 +\frac{A q^2+B}{x} + q^3 - \frac{1}{q}.\eeq 
In the sequel, we shall refer to this equation as \PIII\ since it is the generic case.

Consider the Hamiltonian associated with \PIII\ \eqref{eq:PIII} given by 
\beq \label{eq:PT.HM3p.DE30}
{\HIII}(q,p,x;a,b,\ep) = q^2p^2-xq^2p-(2a+2b+1)qp+\ep xp +2bxq,
\eeq 
with $a$ and $b$ parameters and $\ep=\pm1$, see \cite{refJM81,refOkamotoPIII}. 
Then $p(x)$ and $q(x)$ satisfy the Hamiltonian system
\begin{subequations}\label{eq:PT.HM3.DE32}\begin{align}
x\deriv{q}{x}&=\pderiv{\HIII}{p}\phantom{-}=2q^2p-xq^2-(2a+2b+1)q+\ep x, \label{eq:PT.HM3.DE32a} \\
x\deriv{p}{x}&=-\pderiv{\HIII}{q}=-2qp^2+2xqp + (2a+2b+1)p -2bx. \label{eq:PT.HM3.DE32b}
\end{align}\end{subequations}
Solving \eqref{eq:PT.HM3.DE32a} for $p(x)$ gives
\[ p(x)={\frac{1}{2q^2}}\left\{x\deriv{q}{x}+xq^2+(2a+2b+1)q-\ep x\right\}, \]
and then substituting this in \eqref{eq:PT.HM3.DE32b} gives
\beq\deriv[2]{q}{x} = \frac1q \left(\deriv{q}{x}\right)^{\!2}- \frac{1}{x} \deriv{q}{x}
 +\frac{2(a-b)q^2}{x} + \frac{2\ep(a+b+1)}{x} + q^3 - \frac{1}{q}.\label{eq:p31}\eeq
which is \PIII\ \eqref{eq:PIII}, 
with parameters \beq\label{p3:params}A=2(a-b),\qquad B=2\ep(a+b+1).\eeq
Solving {\eqref{eq:PT.HM3.DE32b}} for $q(x)$ gives
\[ q(x)=\frac{1}{2p(x-p)}\left\{x\deriv{p}{x}-(2a+2b+1)p+2bx \right\},\]
and then substituting this in \eqref{eq:PT.HM3.DE32a} gives
\beq\deriv[2]{p}{x} =
\frac12\left(\frac{1}{p}+\frac{1}{p-x}\right) \left(\deriv{p}{x}\right)^2-{\frac {p}{x(p-x)}}\deriv{p}{x}
+2\ep p-\frac {2b^2}{p} -\frac {4a^2-1}{2(p-x)}+\frac {1-4(a^2-b^2)-4\ep p^2}{2x}.\label{eq:p31p}\eeq
Then making the transformation 
\beq\label{tr:pw} p(x)=\frac{2\sqrt{z}\,w(z)}{w(z)-1},\qquad x=2\sqrt{z},\eeq
\comment{\beq p(x)=\frac{xu(x)}{u(x)-1},\label{tr:pw1} \eeq
gives
\beq\deriv[2]{u}{x}= \left(\frac{1}{2u} + \frac{1}{u-1}\right)\left(\deriv{u}{x}\right)^2 - \frac{1}{x}\,\deriv{u}{x} + 
\frac{2(u-1)^2(a^2 u^2 -b^2)}{x^2u}+{2\ep u},\eeq
and finally letting
\beq u(x)=w(z)\qquad x=\sqrt{2z},\label{tr:pw2} \eeq}%
in \eqref{eq:p31p} gives
\beq\deriv[2]{w}{z}= \left(\frac{1}{2w} + \frac{1}{w-1}\right)\left(\deriv{w}{z}\right)^2 - \frac{1}{z}\,\deriv{w}{z} + 
\frac{(w-1)^2(a^2 w^2 -b^2)}{2z^2w}+\frac{\ep w}{z},\label{eq:degPV1}\eeq
which is deg-\PV\ \eqref{eq:degPV}
with parameters \beq\a=\tfrac12a^2,\qquad \b=-\tfrac12b^2,\qquad \ga=\ep.\label{params:degPV1}\eeq
Hence we have 
the following result; see also \cite[Theorem 34.2]{refGLS}.

\begin{lemma}{If $q(x)$ is a solution of \eqref{eq:p31} then
\beq\label{tr:qw} w(z)=\frac{xq'(x)+xq^2(x)+(2a+2b+1)q(x)-\ep x}{xq'(x)-xq^2(x)+(2a+2b+1)q(x)-\ep x},\qquad z=\tfrac12x^2,\eeq
with $'\equiv \d/\d x$, is a solution of \eqref{eq:degPV1}, provided that 
\[ x\deriv{q}{x}-xq^2+(2a+2b+1)q-\ep x\not=0.\]
Conversely, if $w(z)$ is a solution of \eqref{eq:degPV1}, then
\beq\label{tr:wq} q(x)=\frac{1}{2\sqrt{z}\,w}\left\{z\deriv{w}{z} +(w-1)(aw+b)\right\},\qquad x=\sqrt{2z},\eeq
 is a solution of \eqref{eq:p31}.}\end{lemma}
\begin{proof} Solving \eqref{eq:PT.HM3.DE32a} for $p(x)$, substituting in \eqref{tr:pw} and solving for $w(z)$ gives \eqref{tr:qw}.
Also solving \eqref{eq:PT.HM3.DE32b} for $q(x)$ and substituting \eqref{tr:pw} into the resulting expression gives \eqref{tr:wq}.
\end{proof}

An alternative method of deriving solutions of \eqref{eq:degPV1} involves the second-order, second-degree equation satisfied associated with the 
Hamiltonian system \eqref{eq:PT.HM3.DE32}, due to Jimbo and Miwa \cite{refJM81} and Okamoto \cite{refOkamotoPIII}, which is often called the ``$\sigma$-equation".
\begin{theorem}{\label{thm:22}If  ${\HIII}(q,p,x;a,b,\ep)$ is given by \eqref{eq:PT.HM3p.DE30}, then
\beq \label{ham:PIII}
\sigma(x;a,b,\ep)=\HIII(q,p,x;a,b,\ep)+qp-\tfrac12\ep x^2+(a+b)^2, 
\eeq 
where $q(x)$ and $p(x)$ satisfy the system \eqref{eq:PT.HM3.DE32},
satisfies the second-order, second-degree equation $(\sPIII)$
\begin{align}\label{eq:PT.HM3.DE34}
\left(x\deriv[2]{\sigma}{x} -\deriv{\sigma}{x}\right)^{\!2} + 2\left\{\left(\deriv{\sigma}{x}\right)^{\!2}-x^2\right\}\left(x\deriv{\sigma}{x}-2\sigma\right) &
-8\ep(a^2-b^2) x\deriv{\sigma}{x}=8(a^2+b^2)x^2.
\end{align} 
Conversely, if $\sigma(x;a,b,\ep)$ satisfies \eqref{eq:PT.HM3.DE34} then the solution of the Hamiltonian system 
\eqref{eq:PT.HM3.DE32} is given by
\begin{subequations}\begin{align} q(x)&=\frac{\ds \ep x\sigma''(x)-\ep(2a+2b+1)\sigma'(x)-2(a-b) x}{\ds x^2-\left[\sigma'(x)\right]^{2}},\\
p(x)&=\tfrac12\ep\sigma'(x)+\tfrac12x.  \label{tr:sigma-pq}\end{align}\end{subequations}
}\end{theorem}
\begin{proof}See Jimbo and Miwa \cite{refJM81} and Okamoto \cite{refOkamotoPIII}.\end{proof}


Consequently solutions of deg-\PV\ \eqref{eq:degPV1} can be expressed in terms of solutions of $\sPIII$ \eqref{eq:PT.HM3.DE34}. 
\begin{corollary}{\label{coro29}If $\sigma(x;a,b,\ep)$ is a solution of $\sPIII$ \eqref{eq:PT.HM3.DE34}, then
\beq w(z;a,b,\ep)=\frac{\sigma'(x;a,b,\ep)+\ep x}{\sigma'(x;a,b,\ep)-\ep x}, \qquad z=\tfrac12x^2,\eeq
is a solution of \eqref{eq:degPV1}.
}\end{corollary}
\begin{proof}This immediately follows from \eqref{tr:pw} and  Theorem \ref{thm:22}.\end{proof}

\begin{remark}{\rm From Lemma \ref{lem21} and Corollary \ref{coro29}, it's clear that it's simpler to derive solutions of deg-\PV\ \eqref{eq:degPV1} from
equation \eqref{PV0-SIII} rather than equation \eqref{tr:qw}. Further as shown in \S\ref{sec:p3sol}, classical solutions of  $\sPIII$ 
involve one determinant, whereas classical solutions of \PIII\ involve two determinants.
}\end{remark}

\section{\label{sec:p3sol}Classical solutions of \PIII\ and \sPIII}

\subsection{\label{sec:p3sol1}Rational solutions of \PIII\ and \sPIII}
Rational solutions of \PIII\ \eqref{eq:PIII} are classified in the following theorem.
\begin{theorem}\label{thm:P3rats}
Equation \eqref{eq:PIII} has a rational solution if and only if 
\beq\ep_1A+\ep_2B=4n,\eeq with $n\in\Z$ and $\ep_1^2=1$, $\ep_2^2=1$, independently.
\end{theorem}
\begin{proof}For details see \luk\ \cite{refLuk67b}; see also \cite{refMCB97,refMurata95}.\end{proof}

Umemura \cite{refUmemura20}\footnote{The original manuscript was written by Umemura in 1996 for the proceedings of the conference ``\textit{Theory of nonlinear special functions: the \p\ transcendents}" in Montreal, which were not published; for further details see \cite{refOkOh20}.}
derived special polynomials associated with rational solutions of \PIII\ \eqref{eq:PIII}, which we now define; see also \cite{refPAC03,refKajiwara,refKM99}. 

\begin{defn}{\rm The \textit{Umemura polynomial} $S_{n}(x;\mu)$ is given by the recursion relation 
\beq \label{eq:Srecrel}
S_{n+1}S_{n-1} =-x\left\{S_{n}\deriv[2]{S_{n}}{x}-\left(\deriv{S_{n}}{x}\right)^{\!2}\right\} -S_{n} \deriv{S_{n}}{x}
+(x+\mu)S_{n}^2,\eeq 
where $S_{-1}(x;\mu)=S_0(x;\mu)=1$, with $\mu$ an arbitrary parameter. 
}\end{defn}
 
\begin{remarks}{\rm 
\begin{enumerate}[(i)] \item[]
\item The Umemura polynomial $S_n(x;\mu)$ has the Wronskian representation
\beq\label{def:Umpoly} S_n(x;\mu) = 
c_n\Wr\left(\ph_1,\ph_3,\ldots,\ph_{2n-1}\right), 
\qquad c_n=\prod_{k=0}^n(2k+1)^{n-k},\eeq
 where $\Wr(\ph_1,\ph_2,\ldots,\ph_n)$ is the Wronskian defined by
\[\Wr(\ph_1,\ph_2,\ldots,\ph_n) = 
\left|\begin{matrix}
\ph_1 & \ph_2 & \ldots & \ph_n\\
\ph_1' & \ph_2' & \ldots & \ph_n'\\
\vdots & \vdots & \ddots & \vdots \\
\ph_1^{(n-1)} & \ph_2^{(n-1)} & \ldots &
\ph_n^{(n-1)} \end{matrix}\right|,\qquad \ph_j^{(k)}=\deriv[k]{\ph_j}{x},
\]
and
\[\ph_m(x;\mu)=L_{2m-1}^{(\mu-2m+1)}(-x),\]
with $L_k^{(\a)}(x)$ the Laguerre polynomial, for details see Kajiwara and Masuda \cite{refKM99}; see also \cite{refPAC03,refKajiwara}. 
\item Rational solutions of \PIII\ \eqref{eq:PIII} are expressed in terms of Umemura polynomials.
For example, 
\beq \label{PIII:rat} w_n(z;\mu)=1+\deriv{}{z}\ln \left\{\frac{ S_{n-1}(z;\mu-1)}{ S_n(z;\mu)}\right\},\eeq
satisfies \eqref{eq:PIII} for the parameters
\[A=2n+2\mu-1,\qquad B=2n-2\mu+1,\]
for details see \cite{refPAC03,refKajiwara,refKM99}.
\comment{\item Bothner, Miller and Sheng \cite{refBM,refBMS} studied the asymptotics of the (scaled) poles and roots of the rational solutions in their so-called ``eye-problem".
Bothner, Miller and Sheng \cite{refBM,refBMS} 
study numerically how the distributions of poles and zeros of the rational solutions of \PIII\ \eqref{eq:PIII} behave as $n$ increases and how the patterns vary with $\mu\in\mathbb{C}$ (note that they use a different notation to our notation).}
\end{enumerate}}\end{remarks}

To describe rational solutions of deg-\PV\ \eqref{eq:degPV}, it is more convenient to use rational solutions of $\sPIII$ \eqref{eq:PT.HM3.DE34}, which involve one Umemura polynomial and are discussed in the following theorem, whereas rational solutions of \PIII\ \eqref{eq:PIII}  which involve two Umemura polynomials, as shown in \eqref{PIII:rat}.

\begin{theorem}
The rational function solution of $\sPIII$ \eqref{eq:PT.HM3.DE34} is given by
\begin{subequations}\label{sol:sigma_rat} \beq 
\sigma_n(x;\mu,\ep) = 2x\deriv{}{x} \left\{\ln S_{n}(x;\mu)\right\} - \tfrac12 x^2-2\mu x-\tfrac14,\qquad n\geq0,
\label{sol:sigma1}\eeq 
with $S_{n}(x;\mu)$ the Umemura polynomial, for the parameters
\beq \label{sig:params1} a=n+\tfrac12,\qquad b=\mu,\qquad \ep=1. 
\eeq \end{subequations}
\end{theorem}
\begin{proof}See Clarkson  \cite{refPAC03}.
\end{proof}
\subsection{\label{sec:p3sol2}Special function solutions of \PIII\ and \sPIII}
Special function solutions of \PIII\ \eqref{eq:PIII}, which are expressed in terms of Bessel functions and are classified in the following Theorem.

\begin{theorem}{\label{thm:P3sf} 
Equation \eqref{eq:PIII} has solutions expressible in terms of the Riccati equation 
\beq \label{eq:PT.SF.eq31}
x\deriv{q}{x}=\ep_1xq^2+(A\ep_1-1)q+\ep_2x,
\eeq if and only if 
\beq \ep_1A+\ep_2B=4n+2,\eeq 
with $n\in\Z$ and $\ep_1^2=1$, $\ep_2^2=1$, independently.
Further, the Riccati equation \eqref{eq:PT.SF.eq31} has the solution 
\beq \label{eq:PT.SF.eq32}
q(x)=-{\ep_1}\,\deriv{}{x}\ln\psi_{\nu}(x),
\eeq 
where $\psi_{\nu}(x)$ satisfies
\beq \label{eq:PT.SF.eq33}x\deriv[2]{\psi_{\nu}}{x}+(1-2\ep_1\nu)\deriv{\psi_{\nu}}{x}+\ep_1\ep_2x\psi_{\nu}=0,\eeq 
which has solution
\beq \label{psinu}
\psi_{\nu}(x)=\begin{cases} x^{\nu}\left\{C_1J_{\nu}(x)+C_2Y_{\nu}(x)\right\},&\mbox{\rm if}\quad \ep_1=1,\phantom{-}\enskip\ep_2=1,\\
x^{-\nu}\left\{C_1J_{\nu}(x)+C_2Y_{\nu}(x)\right\},&\mbox{\rm if}\quad \ep_1=-1,\enskip\ep_2=-1,\\
x^{\nu}\left\{C_1I_{\nu}(x)+C_2K_{\nu}(x)\right\},&\mbox{\rm if}\quad\ep_1=1,\phantom{-}\enskip\ep_2=-1,\\
x^{-\nu}\left\{C_1I_{\nu}(x)+C_2K_{\nu}(x)\right\},\quad&\mbox{\rm if}\quad\ep_1=-1,\enskip\ep_2=1,\end{cases}\eeq 
with $C_1$, $C_2$ arbitrary constants, and $J_{\nu}(x)$, $Y_{\nu}(x)$, $I_{\nu}(x)$, $K_{\nu}(x)$ Bessel functions.
}\end{theorem}

\begin{proof} For details see 
Okamoto \cite{refOkamotoPIII}; see also \cite{refPACreview,refGLS,refMW,refMCB97,refMurata95}.\end{proof}

Determinantal representations of special function solutions of \PIII\ \eqref{eq:PIII} were given by Okamoto \cite{refOkamotoPIII}; see also \cite{refFW02,refMasuda07}.
As for rational solutions, to describe special function solutions of deg-\PV\ \eqref{eq:degPV}, it is more convenient to use special function solutions of $\sPIII$ \eqref{eq:PT.HM3.DE34}, which are discussed in the following theorem.

\begin{theorem}
Suppose $\tau_n(x;\mu,\ep)$ is the determinant given by
\begin{subequations} \label{def:tau} \beq\tau_n(x;\mu,\ep) 
=\det\left[\left(x\deriv{}{x}\right)^{j+k}\ph_{\mu}(x;\ep) \right]_{j,k=0}^{n-1},
\comment{=\left| \begin{matrix} \ph_{\mu} & \ph_{\mu}^{(1)} & \ldots & \ph_{\mu}^{(n-1)}\\
\ph_{\mu}^{(1)} &\ph_{\mu}^{(2)} & \ldots & \ph_{\mu}^{(n)} \\
\vdots & \vdots & \ddots & \vdots \\
\ph_{\mu}^{(n-1)} & \ph_{\mu}^{(n)} & \ldots & \ph_{\mu}^{(2n-2)}
\end{matrix} \right|,\qquad \ph_{\mu}^{(m)} = \left(x\deriv{}{x}\right)^{\!m}\ph_{\mu}(x;\ep)}
\eeq
where
\beq \ph_{\mu}(x;\ep) =\begin{cases} c_1 J_{\mu}(x) + c_2 Y_{\mu}(x), &\quad\text{if}\quad \ep=1,\\
c_1 I_{\mu}(x) + c_2 K_{\mu}(x), &\quad\text{if}\quad \ep=-1, 
\end{cases}\eeq \end{subequations}
with $c_1$, $c_2$ arbitrary constants, and $J_{\mu}(z)$, $Y_{\mu}(z)$, $I_{\mu}(z)$, $K_{\mu}(z)$ Bessel functions.

The Bessel function solution of $\sPIII$ \eqref{eq:PT.HM3.DE34} 
is given by
\begin{subequations}\label{sol:sigma_bessel} 
\beq \sigma_n(x;\mu,\ep) = 2x\deriv{}{x} \left\{\ln \tau_n(x;\mu,\ep)\right\} + \tfrac12\ep x^2+\mu^2-n^2+2n,
\label{sol:sigma}\eeq 
for the parameters
\beq \label{sig:params} a=n,\qquad b=\mu.
\eeq \end{subequations}
\end{theorem}

\begin{lemma}
The determinant $\tau_n(x;\mu,\ep)$ given by \eqref{def:tau} satisfies the equation
\beq \label{eq:taun} x^2\left\{ \tau_n \deriv[2]{\tau_n}{x}-\left(\deriv{\tau_n}{x}\right)^2\right\} + x\tau_n \deriv{\tau_n}{x} = \tau_{n+1}\tau_{n-1}, \eeq
or equivalently the Toda equation
\beq \left(x\deriv{}{x}\right)^{\!2} \ln \tau_n= \frac{\tau_{n+1}\tau_{n-1}}{\tau_n^2}. \eeq
\end{lemma} 
\begin{proof} See Okamoto \cite[Theorem 2]{refOkamotoPIII}.
\end{proof}

\section{\label{sec:bts}\bts}
A \textit{\bt}\ relates the solution of a \peq\ either to another solution of the same equation with different values of the parameters, or to another \peq. All \peqs, except the first \peq, have \bts.
Hierarchies of classical solutions of the \p\ equations can be obtained by applying \bk\ transformations to a ``seed solution". 

Let $w_j(z_jz;\a_j,\b_j,\ga_j)$, $j=0,1,2$, be solutions of deg-\PV\  \eqref{eq:degPV} with
\[\begin{array}{ll} z_1=-z_0,\qquad 
&({\a_1}, {\b_1}, {\ga_1})=(\a_0,\b_0,-\ga_0),\\[2.5pt]
z_2=z_0,\qquad &({\a_2}, {\b_2}, {\ga_2})=(-\b_0,-\a_0,-\ga_0).
\end{array}\]
Then deg-\PV\ \eqref{eq:degPV} has the symmetries 
\begin{align}
\label{degPV:S1}\mathcal{S}_1:\qquad &{w_1}({z_1})={w_0(-z_0)},\\
\label{degPV:S2}\mathcal{S}_2:\qquad &{w_2}({z_2})={1}/{w_0(z_0)}.
\end{align}

\begin{theorem}\label{thm:bt}Suppose that $W_0=w(z;\a,\b,\ga)$ satisfies deg-\PV\ \eqref{eq:degPV} with parameters 
\[ \a=\tfrac12a^2,\qquad \b=-\tfrac12b^2,\qquad \ga=c.\]
Let $W_j=w(z;\a_j,\b_j,\ga_j)$, $j=1,2,3,4$, be solutions of \eqref{eq:degPV} with parameters
\[\begin{array}{l@{\qquad}l@{\qquad}l}
\a_1=\tfrac12(a+1)^2,& \b_1=-\tfrac12b^2,& \ga_1=c,\\[2.5pt]
\a_2=\tfrac12(a-1)^2,& \b_2=-\tfrac12b^2,& \ga_2=c,\\[2.5pt]
\a_3=\tfrac12a^2,& \b_3=-\tfrac12(b+1)^2,& \ga_3=c,\\[2.5pt]
\a_4=\tfrac12a^2,& \b_4=-\tfrac12(b-1)^2,& \ga_4=c,
\end{array}\] respectively. Then these solutions can be obtained from $W_0$ as follows
\begin{subequations}\begin{align}
\W_1:\quad W_1
&=\frac{\ds\left\{zW_0'+(W_0-1) (aW_0 -b)\right\} \left\{zW_0'+(W_0-1) (aW_0 +b)\right\}}{\ds z^2(W_0')^2+2a zW_0(W_0-1)W_0'+2cz^2W_0(W_0-1)+(W_0-1)^{2} (a^2 W_0^2 -b^2)},\\
\W_2:\quad W_2
&=\frac{\ds\left\{zW_0'-(W_0-1) (aW_0 -b)\right\} \left\{zW_0'-(W_0-1) (aW_0 +b)\right\}}{\ds z^2(W_0')^2-2 azW_0(W_0-1)W_0'+2cz^2W_0(W_0-1)+(W_0-1)^{2} (a^2 W_0^2 -b^2)},\\
\W_3:\quad W_3
&=\frac{\ds z^2(W_0')^2+2bz(W_0-1) W_0'+2cz^2 W_0^{2}(W_0-1) -(W_0-1)^{2} (a^2 W_0^2 -b^2)}{\ds\left\{zW_0'-(W_0-1) (aW_0 -b)\right\} \left\{zW_0'+(W_0-1) (aW_0 +b)\right\}},\\
\W_4:\quad W_4
&=\frac{\ds z^2(W_0')^2-2bz(W_0-1) W_0'+2cz^2 W_0^{2}(W_0-1) -(W_0-1)^{2} (a^2 W_0^2 -b^2)}{\ds\left\{zW_0'-(W_0-1) (aW_0 -b)\right\} \left\{zW_0'+(W_0-1) (aW_0 +b)\right\}}.\end{align}\end{subequations}
\end{theorem}
\begin{proof} See Adler \cite{refAdler}; also Filipuk and Van Assche \cite{refFVA13}.
\end{proof}

\section{\label{sec:dpvsol}Classical solutions of deg-\PV}
To discuss classical solutions of deg-\PV\ \eqref{eq:degPV}, it is convenient to make the transformation
\beq\label{tr41} w(z)=u(x),\qquad z=\tfrac12x^2,\eeq
in \eqref{eq:degPV}, which gives
\beq\deriv[2]{u}{x}= \left(\frac{1}{2u} + \frac{1}{u-1}\right)\left(\deriv{u}{x}\right)^{\!2} -
\frac{1}{x}\deriv{u}{x} + \frac{4(u-1)^2(\a u^2+\b)}{x^2u} + {2\ga u}.
\label{eq:degPV0}\eeq
We could fix the parameter $\ga$ in \eqref{eq:degPV0}, by rescaling $x$ if necessary, but it is more convenient not to do so. Instead classical solutions will be classified for $\ga=\pm1$.
From Corollary \ref{coro29} and \eqref{tr41}, we have that if $\sigma(x;a,b,\ep)$ is a solution of $\sPIII$ \eqref{eq:PT.HM3.DE34}, then
\beq\label{p50sigma} u(x;a,b,\ep)=\frac{\sigma'(x;a,b,\ep)+\ep x}{\sigma'(x;a,b,\ep)-\ep x},\eeq
is a solution of \eqref{eq:degPV0} with $\ga=\ep$. As remarked above, it is simpler to derive classical solutions of deg-\PV\ \eqref{eq:degPV} from $\sPIII$ rather than \PIII, compare equations \eqref{PV0-SIII} and \eqref{tr:qw}.

\def\u{u_0}
\def\ux{\u'}
\begin{theorem}\label{thm:btu}Supppose that $u_0=u(x;\a,\b,\ga)$ satisfies \eqref{eq:degPV0} with parameters 
\[ \a=\tfrac12a^2,\qquad \b=-\tfrac12b^2,\qquad \ga=c.\]
Let $u_j=u(x;\a_j,\b_j,\ga_j)$, $j=1.2.3,4$ be solutions of \eqref{eq:degPV0} with parameters
\[\begin{array}{l@{\qquad}l@{\qquad}l}
\a_1=\tfrac12(a+1)^2,& \b_1=-\tfrac12b^2,& \ga_1=c,\\[2.5pt]
\a_2=\tfrac12(a-1)^2,& \b_2=-\tfrac12b^2,& \ga_2=c,\\[2.5pt]
\a_3=\tfrac12a^2,& \b_3=-\tfrac12(b+1)^2,& \ga_3=c,\\[2.5pt]
\a_4=\tfrac12a^2,& \b_4=-\tfrac12(b-1)^2,& \ga_4=c,
\end{array}\] respectively.  Then these can be obtained from $u_0$ as follows
\begin{subequations}\begin{align}
\mathcal{U}_1:\quad u_1
&=\frac{\ds\left\{x \ux+2(\u-1) (a\u -b)\right\} \left\{x \ux+2(\u-1) (a\u +b)\right\}}{\ds x^2(\ux)^2+4ax\u(\u-1)\ux+4c\u(\u-1) x^{2}+4(\u-1)^{2} (a^2 \u^2 -b^2)},\\ 
\mathcal{U}_2:\quad u_2
&=\frac{\ds\left\{x \ux-2(\u-1) (a\u -b)\right\} \left\{x \ux-2(\u-1) (a\u +b)\right\}}{\ds x^2(\ux)^2-4 ax\u(\u-1)\ux+4c\u(\u-1) x^{2}+4(\u-1)^{2} (a^2 \u^2 -b^2)},\\ 
\mathcal{U}_3:\quad u_3
&=\frac{\ds x^2(\ux)^2+4 bx (\u-1)\ux+4 c x^2 \u^{2}(\u-1)-4(\u-1)^{2} (a^2 \u^2 -b^2)}{\ds\left\{x \ux-2(\u-1) (a\u -b)\right\} \left\{x \ux+2(\u-1) (a\u +b)\right\}},\\ 
\mathcal{U}_4:\quad u_4
&=\frac{\ds x^2(\ux)^2-4 bx (\u-1)\ux+4 c x^2 \u^{2}(\u-1)-4(\u-1)^{2} (a^2 \u^2 -b^2)}{\ds\left\{x \ux-2(\u-1) (a\u -b)\right\} \left\{x \ux+2(\u-1) (a\u +b)\right\}}.
\end{align}\end{subequations}
satisfy \eqref{eq:degPV} with parameters
\[\begin{array}{l@{\qquad}l@{\qquad}l}
\a_1=\tfrac12(a+1)^2,& \b_1=-\tfrac12b^2,& \ga_1=c,\\[2.5pt]
\a_2=\tfrac12(a-1)^2,& \b_2=-\tfrac12b^2,& \ga_2=c,\\[2.5pt]
\a_3=\tfrac12a^2,& \b_3=-\tfrac12(b+1)^2,& \ga_3=c,\\[2.5pt]
\a_4=\tfrac12a^2,& \b_4=-\tfrac12(b-1)^2,& \ga_4=c,
\end{array}\] respectively.
\end{theorem}

\comment{\begin{theorem}Supppose that $u=u(x;\a,\b,\ga)$ satisfies \eqref{eq:degPV0} with parameters 
\[ \a=\tfrac12a^2,\qquad \b=-\tfrac12b^2,\qquad \ga=c.\]
Then $u_j=u(x;\a_j,\b_j,\ga_j)$ given by
\begin{subequations}\begin{align}
\mathcal{U}_1:\quad u_1
&=\frac{\ds\left\{x \ux+2(\u-1) (a\u -b)\right\} \left\{x \ux+2(\u-1) (a\u +b)\right\}}{\ds x^2(\ux)^2+4ax\u(\u-1)\ux+4cu(\u-1) x^{2}+4(\u-1)^{2} (a^2 \u^2 -b^2)},\\
\mathcal{U}_2:\quad u_2
&=\frac{\ds\left\{x \ux-2(\u-1) (a\u -b)\right\} \left\{x \ux-2(\u-1) (a\u +b)\right\}}{\ds x^2(\ux)^2-4 ax\u(\u-1)\ux+4cu(\u-1) x^{2}+4(\u-1)^{2} (a^2 \u^2 -b^2)},\\
\mathcal{U}_3:\quad u_3
&=\frac{\ds x^2(\ux)^2+4 bx (\u-1)\ux+4 c x^2 u^{2}(\u-1)-4(\u-1)^{2} (a^2 \u^2 -b^2)}{\ds\left\{x \ux-2(\u-1) (a\u -b)\right\} \left\{x \ux+2(\u-1) (a\u +b)\right\}},\\
\mathcal{U}_4:\quad u_4
&=\frac{\ds x^2(\ux)^2-4 bx (\u-1)\ux+4 c x^2 u^{2}(\u-1)-4(\u-1)^{2} (a^2 \u^2 -b^2)}{\ds\left\{x \ux-2(\u-1) (a\u -b)\right\} \left\{x \ux+2(\u-1) (a\u +b)\right\}},
\end{align}\end{subequations}
satisfy \eqref{eq:degPV0} with parameters
\[\begin{array}{l@{\qquad}l@{\qquad}l}
\a_1=\tfrac12(a+1)^2,& \b_1=-\tfrac12b^2,& \ga_1=c,\\[2.5pt]
\a_2=\tfrac12(a-1)^2,& \b_2=-\tfrac12b^2,& \ga_2=c,\\[2.5pt]
\a_3=\tfrac12a^2,& \b_3=-\tfrac12(b+1)^2,& \ga_3=c,\\[2.5pt]
\a_4=\tfrac12a^2,& \b_4=-\tfrac12(b-1)^2,& \ga_4=c,
\end{array}\] respectively.
\end{theorem}}
\begin{proof}This is easily proved by applying \eqref{tr41} to the \bts\ in Theorem \ref{thm:bt}.
\end{proof}

\subsection{\label{sec:dpvsol1}Algebraic solutions}
Since deg-\PV\ \eqref{eq:degPV} and equation \eqref{eq:degPV0} are related by the transformation \eqref{tr41} then
algebraic solutions of deg-\PV\ \eqref{eq:degPV}, which are rational functions of $\sqrt{z}$, are equivalent to rational solutions of \eqref{eq:degPV0}, which are rational functions of $x$. Therefore we discuss rational solutions of \eqref{eq:degPV0}, which are classified in the following Theorem.

\begin{theorem}{Necessary and sufficient conditions for the existence of
rational solutions of \eqref{eq:degPV0} are either
\beq \label{eq:pvalg1}
(\a,\b,\ga)=\left(\tfrac12(n+\tfrac12),-\tfrac12\mu^2,1\right),\eeq or
\beq \label{eq:pvalg2}
(\a,\b,\ga)=\left(\tfrac12\mu^2,-\tfrac12(n+\tfrac12),-1\right),\eeq where
$n\in\Z$ and $\mu$ is an arbitrary constant.}\end{theorem}

\begin{proof} For details see Gromak, Laine and Shimomura \cite[\S38]{refGLS}; see also \cite{refMCB97,refMurata95}. \end{proof}

We remark that the solutions of \eqref{eq:degPV0} satisfying \eqref{eq:pvalg1} are related to those satisfying \eqref{eq:pvalg2} through the analog of the symmetry \eqref{degPV:S2}.
\comment{\[
\widetilde{\mathcal{S}}:\qquad \widetilde{u}(\widetilde{x})
=\ifrac{1}{u(x)},\qquad \widetilde{x}=x, \qquad
(\widetilde{\a},\widetilde{\b},\widetilde{\ga})
=(-\b,-\a,-\ga).\] }%
Consequently we shall be concerned only with rational solutions of \eqref{eq:degPV0} for the parameters given by \eqref{eq:pvalg1}.

\begin{theorem}\label{thm:pvalgsol1}{The rational solution of \eqref{eq:degPV0} for the parameters 
\eqref{eq:pvalg1} is given by
\beq \label{eq:pvalgsol1}
u_n(x;\mu)
=1-\frac{x S_{n}^2(x;\mu)}{S_{n+1}(x;\mu)S_{n-1}(x;\mu)},\qquad n\geq0, \eeq 
where $S_n(x;\mu)$ is the Umemura polynomial \eqref{def:Umpoly}.
}\end{theorem}

\begin{proof} 
Substituting the rational solution of $\sPIII$ \eqref{eq:PT.HM3.DE34} given by \eqref{sol:sigma_rat} into \eqref{p50sigma} 
and then using the reccurence relation \eqref{eq:Srecrel} gives the result.
\end{proof}

\begin{remark}{\rm The Umemura polynomial $S_n(x;\mu)$ satisfies the difference equation
\beq S_{n+1}(x;\mu)S_{n-1}(x;\mu)=x S_{n}^2(x;\mu) + \mu S_n(x;\mu+1)\,S_n(x;\mu-1).
\label{eq:Sndifference}\eeq
Hence from \eqref{eq:pvalgsol1} there are two alternative representations of the rational solution
\[\begin{split}
u_n(x;\mu)&=\frac{\mu S_n(x;\mu+1)\,S_n(x;\mu-1)}{\mu S_n(x;\mu+1)\,S_n(x;\mu-1)+xS_n^2(x;\mu)},\qquad
u_n(x;\mu)=\frac{\mu S_n(x;\mu+1)\,S_n(x;\mu-1)}{S_{n+1}(x;\mu)S_{n-1}(x;\mu)}.
\end{split}\]
}\end{remark}

\subsection{\label{sec:dpvsol2}Bessel function solutions}
\begin{theorem}{Necessary and sufficient conditions for the existence of Bessel function solutions of \eqref{eq:degPV0} are either
\beq \label{eq:pvbes1}
(\a,\b,\ga)=\left(\tfrac12n^2,-\tfrac12\mu^2,\ep\right),\eeq or
\beq \label{eq:pvbes2}
(\a,\b,\ga)=\left(\tfrac12\mu^2,-\tfrac12n^2,-\ep\right),\eeq 
with $\ep=\pm1$, and
where $n\in\Z^+$ and $\mu$ is an arbitrary constant.}\end{theorem}
\begin{proof} From \eqref{p3:params} and \eqref{params:degPV1}, the parameters in \PIII\  \eqref{eq:PIII} and deg-\PV\ \eqref{eq:degPV0}  are given by
\[ (A,B)=\big(2(a-b),2\ep(a+b+1)\big),\qquad (\a,\b,\ga)= (\tfrac12a^2,-\tfrac12b^2,\ep),\]
respectively, for parameters $a$, $b$ and $\ep$.
The result then follows from Theorem \ref{thm:P3sf}.
\end{proof}

\begin{theorem}\label{thm:BesselPV0}The Bessel function solution of \eqref{eq:degPV0} for the parameters 
\[(\a,\b,\ga)=\left(\tfrac12n^2,-\tfrac12\mu^2,\ep\right),\]
is given by
\beq u_n(x;\mu,\ep)
=1+\frac{\ep x^2\tau_{n}^2(x;\mu,\ep)}{\tau_{n+1}(x;\mu,\ep)\,\tau_{n-1}(x;\mu,\ep)},\qquad n\geq1,\eeq
where
\beq \tau_n(x;\mu,\ep) 
=\det\left[\left(x\deriv{}{x}\right)^{\!j+k}\ph_{\mu}(x;\ep) \right]_{j,k=0}^{n-1},\eeq
and $\tau_0(x;\mu,\ep) =1$,
with
\beq \ph_{\mu}(x;\ep) =\begin{cases} c_1 J_{\mu}(x) + c_2 Y_{\mu}(x), &\quad\text{if}\quad \ep=1,\\
c_1 I_{\mu}(x) + c_2 K_{\mu}(x), &\quad\text{if}\quad \ep=-1,
\end{cases}\eeq
$c_1$ and $c_2$ arbitrary constants, and
$J_{\mu}(x)$, $Y_{\mu}(x)$,  $I_{\mu}(x)$ and $K_{\mu}(x)$ Bessel functions. 
\end{theorem}

\begin{proof} 
Substituting the Bessel function solution of $\sPIII$ \eqref{eq:PT.HM3.DE34} given by \eqref{sol:sigma_bessel} into \eqref{p50sigma} 
and then using \eqref{eq:taun} gives the result.
\end{proof}

\comment{\begin{corollary} An alternative description of the Bessel function solution of \eqref{eq:degPV0} for the parameters \eqref{eq:pvbes1} 
is
\beq u_n(x;\mu,\ep)=1+\frac{\ep x^2}{\delta^{(2)}\left(\tau_{n}(x;\mu,\ep)\right)}
=1+\frac{\ep x^2\tau_{n}^2(x;\mu,\ep)}{\tau_{n+1}(x;\mu,\ep)\,\tau_{n-1}(x;\mu,\ep)},\eeq
where $\tau_{n}(x;\mu,\ep)$ is given by \eqref{def:tau} and $\delta$ is the Euler operator
\[\delta\equiv x\deriv{}{x}.\]
\end{corollary}
\begin{corollary}An alternative description of the Bessel function solution of \eqref{eq:degPV0} for the parameters \eqref{eq:pvbes1} is
\beq u_n(x;\mu,\ep)=\ep \frac{x^2\tau_{n}(x;\mu+1,\ep)\,\tau_{n}(x;\mu-1,\ep)}{\tau_{n+1}(x;\mu,\ep)\,\tau_{n-1}(x;\mu,\ep)}.\eeq
\end{corollary}}

\begin{corollary} {\label{thm:BesseldegPV} The Bessel function solution of \eqref{eq:degPV0} for the parameters 
\[(\a,\b,\ga)=\left(\tfrac12n^2,-\tfrac12\mu^2,2\ep\right),\]
is given by
\beq \label{wn:soln} w_n(z;\mu,\ep)
=1+\frac{\ep z\mathcal{T}_{n}^2(z;\mu,\ep)}{\mathcal{T}_{n+1}(z;\mu,\ep)\,\mathcal{T}_{n-1}(z;\mu,\ep)},\qquad n\geq1,\eeq
where
\beq \mathcal{T}_n(z;\mu,\ep) 
=\det\left[\left(z\deriv{}{z}\right)^{\!j+k}\psi_{\mu}(z;\ep) \right]_{j,k=0}^{n-1},\eeq
and $\mathcal{T}_0(z;\mu,\ep)=1$,
with
\beq \ph_{\mu}(z;\ep) =\begin{cases} c_1 J_{\mu}(2\sqrt{z}) + c_2 Y_{\mu}(2\sqrt{z}), &\quad\text{if}\quad \ep=1,\\
c_1 I_{\mu}(2\sqrt{z}) + c_2 K_{\mu}(2\sqrt{z}), &\quad\text{if}\quad \ep=-1, 
\end{cases}\eeq
$c_1$ and $c_2$ arbitrary constants, and
$J_{\mu}(x)$, $Y_{\mu}(x)$,  $I_{\mu}(x)$ and $K_{\mu}(x)$ Bessel functions. 
}\end{corollary}

\comment{\[u_1(x;\mu,1)=\frac{\ph_{\mu+1}(x;1) \left[x\ph_{\mu+1}(x;1)-2 \mu \ph_{\mu}(x;1)\right]}{x \ph_{\mu+1}^2(x;1)-2\mu \ph_{\mu+1}(x;1) \ph_{\mu}(x;1)+x \ph_{\mu}^2(x;1)}\]
\[u_1(x;\mu,-1)=\frac{x \ph_{\mu+1}^2(x;-1)+2\mu \ph_{\mu+1}(x;-1) \ph_{\mu}(x;-1)-x \ph_{\mu}^2(x;-1)}{\ph_{\mu+1}(x;-1) \left[x\ph_{\mu+1}(x;-1) +2 \mu \ph_{\mu}(x;-1)\right]}\]}

In the next Lemma, it is shown that the first solution $u_1(x;\mu,\ep)$, the ``seed solution", satisfies a first-order, second-degree equation.

\begin{lemma} The solution of \eqref{eq:degPV0} for the parameters 
\[(\a,\b,\ga)=\left(\tfrac12,-\tfrac12\mu^2,\ep\right),\] is
\beq u_1(x;\mu,\ep)=\frac{\ph_{\mu+1}(x;\ep) \left[x\ph_{\mu+1}(x;\ep)-2\ep \mu \ph_{\mu}(x;\ep)\right]}{x \ph_{\mu+1}^2(x;\ep)-2\ep\mu \ph_{\mu+1}(x;\ep) \ph_{\mu}(x;\ep)+\ep x \ph_{\mu}^2(x;\ep)},\label{sol:u1}\eeq
where
\[\ph_{\mu}(x;\ep) =\begin{cases} c_1 J_{\mu}(x) + c_2 Y_{\mu}(x), &\quad\text{if}\quad \ep=1,\\
c_1 I_{\mu}(x) + c_2 K_{\mu}(x), &\quad\text{if}\quad \ep=-1, 
\end{cases}\]
with $c_1$ and $c_2$ constants, satisfies the first-order, second-degree equation
\beq x^2\left(\deriv{u}{x}\right)^{\!2}-4x u (u-1)\deriv{u}{x}+4\ep x^2 u (u-1)+4 (u-1)^{2} (u^2 -\mu^2 )=0.\label{u1:ode1}\eeq
\comment{and 
\[u_1(x;\mu,-1)=\frac{\ph_{\mu+1}(x;-1) \left[x\ph_{\mu+1}(x;-1) +2 \mu \ph_{\mu}(x;-1)\right]}{x \ph_{\mu+1}^2(x;-1)+2\mu \ph_{\mu+1}(x;-1) \ph_{\mu}(x;-1)-x \ph_{\mu}^2(x;-1)}\]
satisfies
\beq x^2\left(\deriv{u}{x}\right)^{\!2}-4x u (u-1) \deriv{u}{x}-4x^2 u (u-1)+4 (u-1)^{2} (u^2 -\mu^2 )=0,\eeq}
\end{lemma}

\begin{proof}
Define \[ \Phi_\mu(x;\ep)= \frac{\ph_{\mu+1}(x;\ep)}{\ph_{\mu}(x;\ep)},\]
then from \eqref{sol:u1}
\beq u_1(x;\mu,\ep)=
1-\frac{x}{\ep x\Phi_\mu^2-2\mu\Phi_\mu+ x},
\label{u1:Phi} \eeq
and $\Phi_\mu(x;\ep)$ satisfies the Riccati equation
\beq x\deriv{\Phi_\mu}{x}=\ep x\Phi_\mu^2- (2\mu+1)\Phi_\mu+x.\label{eq:Phi}\eeq
Next we assume that $u_1(x;\mu,\ep)$ satisfies a first-order, second-degree equation of the form
\beq x^2\left(\deriv{u}{x}\right)^{\!2} + x 
\left[f_2(x,\mu,\ep) u^2 + f_1(x,\mu,\ep) u +f_0(x,\mu,\ep)\right] 
\deriv{u}{x} + \sum_{j=0}^4 g_j(x,\mu,\ep)u^j=0,
\label{eq:u1}\eeq
where $\left\{f_j(x,\mu,\ep)\right\}_{j=0}^2$ and $\left\{g_j(x,\mu,\ep)\right\}_{j=0}^4$ are to be determined. Then substituting \eqref{u1:Phi} into \eqref{eq:u1}, using the fact that $\Phi_\mu(x;\ep)$ satisfies \eqref{eq:Phi}
and equating coefficients of powers of $\Phi_\mu$ yields
\begin{align*}
&f_{2} = -4,&& f_{1} = 4,&& f_{0} = 0,\\ 
&g_4=4,&& g_{3} = -8,&& g_{2} = 4\ep x^{2}-4 \mu^{2}+4,&& g_{1} = -4\ep x^{2}+8 \mu^{2},&&
g_{0} = -4 \mu^{2}.\end{align*}
Hence we obtain equation \eqref{u1:ode1}, as required.
\end{proof}

This demonstrates that special function solutions of \eqref{eq:degPV0}, and hence also deg-\PV\ \eqref{eq:degPV}, are different from special function solutions of \PII--\PVI\ where the ``seed solution" satisfies a Riccati equation, a first-order, first-degree equation. 

\begin{remark}{\rm Gromak, Laine and Shimomura \cite[equation (38.7)]{refGLS} give, without proof, a first-order, second-degree equation associated with Bessel function solutions of deg-\PV\ \eqref{eq:degPV}; see also Filipuk and Van Assche \cite[\S2.3]{refFVA13}.}\end{remark}

\comment{\begin{align*}
&\deriv{}{x}\left\{ x^2\left(\deriv{u}{x}\right)^{\!2}-4x u (u-1)\deriv{u}{x}+4\ep x^2 u (u-1)+4 (u-1)^{2} (u^2 -\mu^2 )\right\}\\
&\qquad -\left\{\frac{3u-1}{u(u-1)}\deriv{u}{x}-\frac{2(u-1)}{x}\right\}
\left\{ x^2\left(\deriv{u}{x}\right)^{\!2}-4x u (u-1)\deriv{u}{x}+4\ep x^2 u (u-1)+4 (u-1)^{2} (u^2 -\mu^2 )\right\}\\
&\qquad\qquad=2x\left\{x\deriv{u}{x}-2u(u-1)\right\}\left\{\deriv[2]{u}{x}- \left(\frac{1}{2u} + \frac{1}{u-1}\right)\left(\deriv{u}{x}\right)^{\!2} +
\frac{1}{x}\deriv{u}{x} - \frac{2(u-1)^2(u^2-\mu^2)}{x^2u} - {2\ep u}\right\}\\
 \end{align*}
 \[\begin{split}
\deriv{}{z}\left(\wz+\ep w^2+\frac{\k w}{z}+1\right)&-\frac{1}{w}\left(\wz+\ep w^2+\frac{\k w}{z}+1\right)^{\!2}
+ \frac{(\k+1)w + 2z}{wz}\left(\wz+\ep w^2+\frac{\k w}{z}+1\right)\\
&=\wzz - \frac{1}{w}\!\left(\wz\right)^{\!2} + \frac{1}{z} \,\wz - \frac{\ep(\k-1) w^2}{z} + \frac{\k+1}{z} - w^3 + \frac{1}{w}
\end{split}\]
with $\ep^2=1$.}

\def\cSG{complex sine-Gordon}
\section{\label{sec:app}Applications}
\subsection{Complex sine-Gordon equation}
Consider the two-dimensional complex sine-Gordon 
equation 
\beq \nabla^2\psi + \frac{(\nabla\psi)^2\overline{\psi}}{1-|\psi|^2}+\psi\left(1-|\psi|^2\right)=0,\label{eq:csg1}\eeq 
where $\nabla\psi=(\psi_x,\psi_y)$. 
Making the transformation $$\psi(x,y)=\cos(\ph(x,y))\exp\{\i\eta(x,y)\},\qquad
\overline{\psi}(x,y)=\cos(\ph(x,y))\exp\{-\i\eta(x,y)\},$$
in the \cSG\ equation \eqref{eq:csg1} yields 
\begin{align*}
&\nabla^2\ph+\frac{\cos\ph}{\sin^3\ph}(\nabla \eta)^2-\tfrac12 \sin(2\ph)=0,\\
&\sin(2\ph)\,\nabla^2\eta=4(\ph_x\eta_x+\ph_y\eta_y),
\end{align*}
which is the \textit{Pohlmeyer-Lund-Regge model}\ \cite{refLund77,refLR,refPohl}.

The \cSG\ equation \eqref{eq:csg1} has a separable solution in polar coordinates given by $\psi(r,\theta)=R_n(r)\,\e^{\i n\theta}$, 
where $R_n(r)$ satisfies
\beq \deriv[2]{R_n}{r}+\frac1r\,\deriv{R_n}{r} + \frac{R_n}{1-R_n^2}\left\{\left(\deriv{R_n}{r}\right)^{\!2}-\frac{n^2}{r^2}\right\}+R_n\big(1-R_{n}^2\big)=0,
\label{eq:csg2}\eeq 
We remark that this equation 
also arises in extended quantum systems \cite{refCFH05,refCH05,refCH08}, in relativity \cite{refGMS} and in coefficients in the three-term recurrence relation for orthogonal polynomials with respect to the weight $w(\theta)=\e^{t\cos\theta}$ on the unit circle, see \cite[equation (3.13)]{refWVAbk}. The orthogonal polynomials for this weight on the unit circle are related to unitary random matrices \cite{refPS}.

Equation \eqref{eq:csg2} can be shown to possess the \p\ property, though it is not in the list of 50 equations given in \cite[Chapter 14]{refInce}.
Equation \eqref{eq:csg2} can be transformed to the fifth \p\ equation \eqref{eq:PVgen} in two different ways.
\begin{enumerate}[(i)]
\item If $R_n(r)$ satisfies \eqref{eq:csg2} then making the transformation 
\beq R_n(r)=\frac{1+u_n(z)}{1-u_n(z)},\qquad r=\tfrac12z,\label{tr:ru}\eeq 
yields
\beq \deriv[2]{u_n}z= \left(\frac{1}{2u_n} + \frac{1}{u_n-1}\right)\left(\deriv{u_n}z\right)^{\!2} -
\frac{1}{z}\,\deriv{u_n}z+\frac{n^2(u_n-1)^2(u_n^2-1)}{8z^2u_n} -\frac{u_n(u_n+1)}{2(u_n-1)}, \label{eq:un} \eeq 
which is \PV\ \eqref{eq:PVgen} 
 with $\alpha=\tfrac18n^2$, $\beta=-\tfrac18n^2$, $\gamma=0$ and $\delta=-\tfrac12$. 
\item If $R_n(r)$ satisfies \eqref{eq:csg2} then making the transformation 
\beq R_n(r)=\frac{1}{\sqrt{1-v_n(x)}},\qquad r=\sqrt{x},\label{tr:rv} \eeq
yields 
\beq \deriv[2]{v_n}{x}=\left(\frac1{2v_n}+\frac1{v_n-1}\right)\left(\deriv{v_n}{x}\right)^{\!2}-\frac1x\,\deriv{v_n}{x}-\frac{n^2(v_n-1)^2}{2x^2 v_n}+\frac{v_n}{2x},\label{eq:vn} \eeq which is deg-\PV\ \eqref{eq:degPV} 
with $\a=0$, $\b=-\tfrac12n^2$ and $\ga=\tfrac12$ 
so is equivalent to \PIII\ \eqref{eq:PIII}, as mentioned above.
\end{enumerate}
This shows that solutions of equations \eqref{eq:un} and \eqref{eq:vn} are related by 
\[ v_n(x) =\frac{4u_n(z)}{1+u^2_n(z)} , \qquad x=\tfrac14z^2.\]

The function $R_n(r)$ satisfies the ordinary differential equation \eqref{eq:csg2},
the differential-difference equations
\begin{subequations}\label{eq23}\begin{align}
&\deriv{R_n}{r}+\frac{n}{r}R_{n}-\big(1-R_{n}^2\big)R_{n-1}=0,\label{eq2}\\
&\deriv{R_{n-1}}{r}-\frac{n-1}{r}R_{n-1}+\big(1-R_{n-1}^2\big)R_{n}=0,\label{eq3}
\end{align}\end{subequations}
since solving \eqref{eq2} for $R_{n-1}(r)$ and substituting in \eqref{eq3} yields equation \eqref{eq:csg2}. 
Also eliminating the derivatives in \eqref{eq23}, after letting $n\to n+1$ in \eqref{eq3}, yields the difference equation
\beq R_{n+1}+R_{n-1}={\frac{2n}{r}\,\frac{R_n}{1-R_n^2}},\label{eq4}
\eeq 
which is known as the discrete \p\ II equation \cite{refNP,refPS}.

If $n=1$ then equations \eqref{eq23} have the solution
\begin{equation*}
R_0(r)=1,\qquad R_1(r)
=\frac{C_1I_1(r)-C_2K_1(r)}{C_1I_0(r)+C_2K_0(r)},
\end{equation*} where $I_0(r)$, $K_0(r)$, $I_1(r)$ and $K_1(r)$ are the imaginary Bessel functions and $C_1$ and $C_2$ are arbitrary constants. For solutions which are bounded at $r=0$ then necesssarily $C_2=0$ and so
\beq \label{ph0ph1}
R_0(r)=1,\qquad R_1(r)
=\frac{I_1(r)}{I_0(r)}.
\eeq 
Hence one can use the difference equation \eqref{eq4} to determine $R_n(r)$, for $n\geq2$, which yields 
\begin{align*}
R_2(r)&=-\frac{rR_1^2(r)+2R_1(r)-r}{r\left[R_1^2(r)-1\right]},\\
R_3(r)&=\frac{R_1^3(r)-rR_1^2(r)-2R_1(r)+r}{R_1(r)\left[rR_1^2(r)+R_1(r)-r\right]},\\
R_4(r)&=\frac{r(r^2+5)R_1^4(r)+4R_1^3(r)-2r(r^2+3)R_1^2(r)+r^3}{r\left[(r^2-1)R_1^4(r)+4rR_1^3(r)-2(r^2+2)R_1^2(r)-4rR_1(r)+r^2 \right]}.
\end{align*}

These results suggest that \eqref{eq:csg2} should be solvable in terms of \PIII\ \eqref{eq:PIII}, which is illustrated in the following theorem.

\begin{theorem}\label{thm:6.1} If $R_n(r)$ satisfies \eqref{eq:csg2} then 
$\ds w_n(r)={R_{n+1}(r)}/{R_{n}(r)}$ satisfies 
\beq \deriv[2]{w_{n}}{r}= \frac{1}{w_n}\left(\deriv{w_{n}}{r}\right)^{\!2}-\frac{1}{r}\,\deriv{w_{n}}{r}
-\frac{2n}{r}w_n^2+ \frac{2(n+1)}{r} + w_n^3-\frac{1}{w_n},\label{p3}\eeq 
which is \PIII\ \eqref{eq:PIII} with parameters $\a=-2n$ and $\b=2(n+1)$. 
\end{theorem}
\begin{proof} See Hisakado \cite{refHisakado} and Tracy \& Widom \cite{refTW99}; see also \cite[\S3.1]{refWVAbk}. \end{proof}

We note that since the parameters in \eqref{p3} satisfy $-\a+\b=4n+2$, with $n\in\Integer^+$, then the equation has solutions expressible in terms of the modified Bessel functions $I_0(r)$ and $I_1(r)$ (as well as $K_0(r)$ and $K_1(r)$, but these are not needed here). 

\begin{theorem}{\label{thm:6.2}Let $\tau_n(r;\nu)$ be the $n\times n$ determinant
\beq \label{def:taun}
\tau_n(r;\nu)=\det\left[\left(r\deriv{}{r}\right)^{j+k}I_{\nu}(r)\right]_{j,k=0}^{n-1},
\comment{\left|\begin{matrix} \psi_\nu & \psi_\nu^{(1)} & \ldots & \psi_\nu^{(n-1)}\\
\psi_\nu^{(1)} & \psi_\nu^{(2)} & \ldots & \psi_\nu^{(n)}\\
\vdots & \vdots & \ddots & \vdots \\
\psi_\nu^{(n-1)} & \psi_\nu^{(n)} & \ldots & \psi_\nu^{(2n-2)}
\end{matrix}\right|,\qquad \psi_\nu^{(k)}=\left(r\deriv{}{r}\right)^k\psi_\nu,}
\eeq 
with $
I_{\nu}(r)$ the {modified Bessel function}, then
\beq 
w_n(r;\nu)=\frac{\tau_{n+1}(r;\nu+1)\,\tau_{n}(r;\nu)}{\tau_{n+1}(r;\nu)\,\tau_{n}(r;\nu+1)}
\equiv\deriv{}{z}\left\{\ln\frac{\tau_{n+1}(z;\nu)}{\tau_{n}(z;\nu+1)}\right\}-\frac{n+\nu}{z},\qquad n\geq0,
\eeq 
satisfies \PIII\ \eqref{eq:PIII} with $\a=2(\nu-n)$ and $\b=2(\nu+n+1)$. 
}\end{theorem}

\begin{proof} See, for example, \cite{refFW02,refMasuda07}. 
\end{proof}

\begin{theorem}{\label{thm:6.3}Equation \eqref{eq:csg2}
has the solution
\beq \label{eq:Rn} R_n(r)=\frac{\tau_{n}(r;1)}{\tau_{n}(r;0)},\eeq 
where $\tau_n(r;\nu)$ is the determinant given by \eqref{def:taun}.}\end{theorem}
\begin{proof} The proof is straightforward using induction. From \eqref{ph0ph1} we have
\[ R_1(r)=\frac{I_1(r)}{I_0(r)}=\frac{\tau_{1}(r;1)}{\tau_{1}(r;0)},\]
so \eqref{eq:Rn} is true if $n=1$.
Assuming \eqref{eq:Rn} holds then from Theorems \ref{thm:6.1} and \ref{thm:6.2} 
\[ R_{n+1}(r)= w_n(r;0) R_n(r) =\frac{\tau_{n+1}(r;1)\,\tau_{n}(r;0)}{\tau_{n+1}(r;0)\,\tau_{n}(r;1)}\times\frac{\tau_{n}(r;1)}{\tau_{n}(r;0)}
= \frac{\tau_{n+1}(r;1)}{\tau_{n+1}(r;0)},\]
as required, and so the result follows by induction.
\end{proof}

\begin{corollary}{\rm Equations \eqref{eq:un} and \eqref{eq:vn} have the Bessel function solutions
\[ u_n(z) = \frac{\tau_n(\tfrac12z;1)+\tau_n(\tfrac12z;0)}{\tau_n(\tfrac12z;1)-\tau_n(\tfrac12z;0)},\qquad
v_n(x) = 1- \frac{\tau_n^2(\sqrt{x};0)}{\tau_n^2(\sqrt{x};1)},\]
respectively, with $\tau_n(r;\nu)$ the determinant given by \eqref{def:taun}.
}\end{corollary}
\begin{lemma}{The formal asymptotic behaviour of the vortex solution $R_n(r)$ is given by
\begin{align}
R_n(r) &= \frac{r^n}{2^n\,n!}\left\{1-\frac{r^{2}}{4(n+1)} + \mathcal{O}\left(r^{4}\right)\right\},&&\mbox{as}\quad r\to0,\\
R_n(r) &= 1-\frac{n}{2r}-\frac{n^2}{8r^2}-\frac{n(n^2+1)}{16r^3}+\mathcal{O}(r^{-4}),&&\mbox{as}\quad r\to\infty.
\end {align}
}\end{lemma}
\begin{proof}
These are determined from \eqref{eq4} and \eqref{ph0ph1}.
\end{proof} 

\figs{\begin{figure}[!ht]\[ \includegraphics[width=3.5in]{figures/cSG1plot1234}\]
\caption{}\end{figure}}%

\figs{\begin{figure}[!ht]\[\begin{array}{c@{\qquad}c}
 \includegraphics[width=2.25in]{figures/cSG1plot1} & \includegraphics[width=2.25in]{figures/cSG1plot2}\\
 R_1(r) & R_2(r)\\
 \includegraphics[width=2.25in]{figures/cSG1plot3} & \includegraphics[width=2.25in]{figures/cSG1plot4} \\
 R_3(r) & R_4(r)
\end{array}\]
\caption{}\end{figure}}

\def\bb{\nu}
\subsection{\label{ssec:gencharlier}Generalised Charlier polynomials}
The {Charlier polynomials} $C_n(k;z)$ 
are a family of orthogonal polynomials introduced in 1905 by Charlier \cite{refChar} given by
\begin{equation}
C_n(k;z) =\HyperpFq{2}{0}\left(-n,-k;;-\ifrac{1}{z}\right)=(-1)^nn!L_n^{(-1-k)}\left(-\ifrac{1}{z}\right),\quad z>0,
\label{charlier}\end{equation}
where $\HyperpFq{2}{0}(a,b;;z)$ is the hypergeometric function 
and $L_n^{(\a)}(z)$ is the {associated Laguerre polynomial}, see, for example, 
\cite[\S18.19]{refDLMF}.
The Charlier polynomials are orthogonal on the lattice $\N$ with respect to the Poisson distribution 
\begin{equation} \w(k) =\frac{z^k}{k!},\qquad z>0,\label{charlierw}\end{equation}
and satisfy the orthogonality condition
\[\sum_{k=0}^\infty C_m(k;z)C_n(k;z)\frac{z^k}{k!}=\frac{n!\,\e^{z}}{z^n}\delta_{m,n}.\]

Smet and Van Assche \cite{refSmetvA} generalized the Charlier weight \eqref{charlierw}
with one additional parameter through the weight function
\[\w(k;\bb)=\frac{\Gamma(\bb+1)\,z^k}{\Gamma(\bb+k+1)\,\Gamma(k+1)},\qquad z>0,\] 
with $\bb$ a parameter such that $\bb>-1$. This gives the discrete weight 
\beq\w(k;\bb) = \frac{z^k}{(\bb+1)_k\,k!},\qquad z>0,\label{gencharlierw}\eeq
where $(\bb+1)_k
=\Gamma(\bb+1+k)/\Gamma(\bb+1)$ is the Pochhammer symbol, on the lattice $\N$.
Discrete orthogonal polynomials are characterized by the discrete Pearson equation
\beq\Delta\big[\sigma(k)\w(k)\big]=\tau(k)\w(k),\label{eq:disPearson}\eeq
where $\Delta$ is the forward difference operator
\[\Delta f(k)=f(k+1)-f(k).\]
The weight \eqref{gencharlierw} satisfies the discrete Pearson equation \eqref{eq:disPearson} with 
\[\sigma(k)=k(k+\bb),\qquad\tau(k)=-k^2-\bb k+z,\]
and so the generalised Charlier polynomials are semi-classical orthogonal polynomials since $\tau(k)$ is a polynomial with deg$(\tau)>1$.
The special case $\nu=0$ was first considered by Hounkonnou, Hounga and Ronveaux \cite{refHHR} and later studied by Van Assche and Foupouagnigni \cite{refVAF03}.  

For the generalised Charlier weight \eqref{gencharlierw}, the orthonormal polynomials $p_n(k;z)$ satisfy the orthogonality condition
\[\sum_{k=0}^\infty p_m(k;z)p_n(k;z)\frac{z^k}{(\bb+1)_k\,k!}=\delta_{m,n},\]
and the three-term recurrence relation
\beq
kp_n(k;z)=a_{n+1}(z)p_{n+1}(k;z)+b_n(z)p_n(k;z)+a_n(z)p_{n-1}(k;z),
\label{deq:rr1}\eeq
with $p_{-1}(k;z)=0$ and $p_0(k;z)=1$. Our interest is in 
the coefficients $a_n(z)$ and $b_n(z)$ in the recurrence relation \eqref{deq:rr1}.

Smet and Van Assche \cite[Theorem 2.1]{refSmetvA} proved the following theorem for the recurrence relation coefficients associated with the generalised Charlier weight \eqref{gencharlierw}.

\begin{theorem}{The recurrence relation coefficients $a_n(z)$ and $b_n(z)$ for orthonormal polynomials associated with the generalised Charlier weight \eqref{gencharlierw} on the lattice $\N$ satisfy the discrete system
\beq\label{SVAsys}\begin{array}{l}
(a_{n+1}^2-z)(a_{n}^2-z)=z(b_n-n)(b_n-n+\bb),\\[5pt] \ds b_n+b_{n-1}-n+\bb+1=\ifrac{nz}{a_n^2},
\end{array}\eeq
with initial conditions
\beq \label{SVAsysic}
a_0^2=0,\qquad b_0=\frac{\sqrt{z}\,\BesselI{\bb+1}(2\sqrt{z})}{\BesselI{\bb}(2\sqrt{z})}=z\deriv{}{z}\big\{\ln\BesselI{\bb}(2\sqrt{z})\big\}-\frac{\bb}{2},
\eeq
with $\BesselI{\nu}(k)$ the {modified Bessel function}.}\end{theorem}

\begin{remark}{\rm The discrete system such as \eqref{SVAsys} for recurrence relation coefficients is sometimes known as the \textit{Laguerre-Freud equations}, cf.~\cite{refBelRon,refHHR,refMagnus86}. 
}\end{remark}
The recurrence relation coefficients $a_n(z)$ and $b_n(z)$ also satisfy the Toda lattice, cf.~\cite[Theorem 3.8]{refWVAbk}
\begin{subequations}\label{eq:toda}
\begin{align}\label{eq:toda1}
&z\deriv{}{z}a_{n}^2=a_{n}^2(b_{n}-b_{n-1}),\\
&z\deriv{}{z}b_n=a_{n+1}^2-a_{n}^2.\label{eq:toda2}
 \end{align}\end{subequations}

Letting $a_n^2(z)=x_n(z)$ and $b_n(z)=y_n(z)$ in \eqref{SVAsys} and \eqref{eq:toda} yields
\begin{align*}
&(x_{n+1}-z)(x_{n}-z)=t(y_n-n)(y_n-n+\bb), &&z\deriv{x_n}{t}=x_{n}(y_n-y_{n-1}), \\ & y_n+y_{n-1}-n+\bb+1=\frac{nz}{x_n},
&&z\deriv{y_n}{z}=x_{n+1}-x_{n}.
\end{align*}
Eliminating $x_{n+1}$ and $y_{n-1}$ in these equations yields the differential system
\begin{subequations}\begin{align}
z\deriv{x_n}{z}&=x_n(2y_n+\bb-n+1)-nz,\label{charsys1a}\\
z\deriv{y_n}{z}&= -x_n+z+\frac{(y_n-n)(y_n-n+\bb)z}{x_n-z}.\label{charsys1b}
\end{align}\end{subequations}
Solving \eqref{charsys1a} for $y_n$ gives
\[ y_n=\frac{z}{2x_n}\deriv{x_n}{z}+\frac{nz}{2x_n}+\frac{n-\bb-1}2,\]
and substituting this into \eqref{charsys1b} yields
\begin{align}
\deriv[2]{x_n}{z}
=\frac12\left(\frac{1}{x_n}+\frac{1}{x_n-z}\right)- \frac{x_n}{z(x_n-z)}\deriv{x_n}{z}
&-\frac{2x_n^2}{z^2}+\frac{4x_n+n^2-\bb^2+1}{2z}-\frac{n^2}{2x_n}+\frac{1-\bb^2}{2(x_n-z)}.
\label{eq:xn}\end{align}
Making the transformation
\beq x_n(z)=\frac{z}{1-w_n(z)}.\label{tr:xw}\eeq
in \eqref{eq:xn} yields
\begin{align}
\deriv[2]{w_n}{z} &= \left(\frac{1}{2w_n} + \frac{1}{w_n-1}\right)\! \left(\deriv{w_n}{z} \right)^{\!2} - \frac{1}{z} \deriv{w_n}{z} 
+\frac{(w_n-1)^2(n^2 w_n^2 -\bb^2)}{2w_nz^2} 
- \frac{2w_n}{z},\label{eq:PT.DE.PV0}\end{align} 
which is deg-\PV\ \eqref{eq:degPV} with parameters $\a=\tfrac12n^2$, $\b=-\tfrac12{\bb}^2$ and $\gamma=-2$.

Solving \eqref{charsys1b} for $x_n$ gives
\beq x_n=-\tfrac12z\deriv{y_n}{z}+z+\tfrac12X_n,\label{eq:charx}\eeq
where
\beq X_n^2= z^2\left(\deriv{y_n}{z}\right)^{\!\!2}+4z(y_n-n)(y_n-n+\bb).\label{eq:charu}\eeq
From \eqref{eq:charu} we get
\begin{align} \deriv{X_n}{z}& = \frac{z^2}{X_n}\deriv[2]{y_n}{z}\deriv{y_n}{z} 
+\frac{z}{X_n}\left(\deriv{y_n}{z}\right)^{\!\!2} + \frac{2z(2y_n-2n+\bb)}{X_n} \deriv{y_n}{z} 
+\frac{2(y_n-n)(y_n-n+\bb)}{X_n}.
\label{eqcharut}\end{align}
Substituting \eqref{eq:charx} into \eqref{charsys1a}, then using \eqref{eqcharut}, solving for $X_n$, and substituting into \eqref{eq:charu} yields the second-order, second-degree equation
\beq \label{eq:yn}
\left(2z\deriv[2]{y_n}{z}+\deriv{y_n}{z}+8y_n-8n+4\bb\right)^{\!\!2}=\frac{(4y_n-2n+2\bb+1)^2}{z}\left\{z\left(\deriv{y_n}{z}\right)^{\!\!2}+4(y_n-n)(y_n-n+\bb)\right\}.\eeq
Making the transformation
\[ y_n(z)=\tfrac12 v_n( x) +\tfrac12n-\tfrac12{\bb}-\tfrac14,\qquad  x=2\sqrt{z},\] 
in \eqref{eq:yn} yields
\beq \label{eq:u2n}
\left(\deriv[2]{v_n}{ x}+4v_n-4n-2\right)^{\!\!2}=\frac{4v_n^2}{ x^2}\left\{\left(\deriv{v_n}{ x}\right)^{\!\!2}+4v_n^2-4(2n+1)v_n+(2n+1)^2-4\bb^2\right\}.\eeq

Equation (A.5) in \cite{refCosgrove06b} is 
\beq \label{eq:A5}
\left(\deriv[2]{v}{ x}-av-b\right)^{\!\!2}=\frac{4v^2}{ x^2}\left\{\left(\deriv{v}{ x}\right)^{\!\!2}-av^2-2bv-c\right\},\eeq
with 
$a$, $b$ and $c$ parameters, an equation derived by Chazy \cite{refChazy09}, 
and is the primed version of equation SD-III in \cite{refCS}.
Hence equation \eqref{eq:u2n} is the special case of equation \eqref{eq:A5} with
\[ 
a=-4, \qquad b=4n+2,\qquad c=4 \bb^{2}-(2n+1)^{2}.\]
Cosgrove \cite{refCosgrove06b} showed that equation \eqref{eq:A5} is solvable in terms of solutions of \PIII\ \eqref{eq:PIII}. Consequently, the solution of \eqref{eq:u2n} is given by
\[ v_n(x)=\frac{x}{2q}\left(\deriv{q}{x}+q^2+1\right),\]
where $q(x)$ satisfies \PIII\ \eqref{eq:PIII} for the parameters $A = 2 \bb -2 n -2$ and $B = 2 \bb +2 n$.

\comment{\begin{theorem}If we define
\beq
\tau_n(z;\bb)=\det\left[\left(z\deriv{}{z}\right)^{\!j+k} \Big(t^{-\bb/2}\BesselI{\bb}\big(2\sqrt{z}\big)\Big)\right]_{j,k=0}^{n-1},
\label{genC:tau}\eeq
with $\tau_0(z;\bb)=1$ and where $\BesselI{\nu}(z)$ is the modified Bessel function,
then the recurrence relation coefficients $a_n(z)$ and $b_n(z)$ have the solutions
\begin{subequations}\label{deq:anbn1}\begin{align}
\label{deq:anbn1a}
a_n^2(z)&=\frac{\tau_{n+1}(z;\bb)\tau_{n-1}(z;\bb)}{\tau_n^2(z;\bb)} = \left(z\deriv{}{z}\right)^2 \ln\tau_n(z;\bb),\\ 
b_n(z) &= z\deriv{}{z}\ln\frac{\tau_{n+1}(z;\bb)}{\tau_n(z;\bb)}.\label{deq:anbn1b}
\end{align}\end{subequations}
\end{theorem}}

\begin{theorem}The recurrence relation coefficients $a_n(z)$ and $b_n(z)$ are given by
\begin{subequations}\label{deq:anbn1}
\begin{align}
\label{deq:anbn1a}
a_n^2(z)&=x_n(z)=\frac{\mathcal{T}_{n+1}(z;\bb)\mathcal{T}_{n-1}(z;\bb)}{\mathcal{T}_n^2(z;\bb)},\\
b_n(z) &=y_n(z)= z\deriv{}{z}\left\{\ln\frac{\mathcal{T}_{n+1}(z;\bb)}{\mathcal{T}_n(z;\bb)}\right\}-\frac{\bb}2,\label{deq:anbn1b}
\end{align}\end{subequations}
where
\[
\mathcal{T}_n(z;\bb)=\det\left[\left(z\deriv{}{z}\right)^{\!j+k} \BesselI{\bb}\big(2\sqrt{z}\big)\right]_{j,k=0}^{n-1},\]
with $\mathcal{T}_0(z;\bb)=1$, and $\BesselI{\nu}(x)$ is the modified Bessel function.
\end{theorem}
\begin{proof}
The expression \eqref{deq:anbn1a} for $a_n^2(z)$ follows immediately by substituting \eqref{wn:soln}
in \eqref{tr:xw}.  To prove the result \eqref{deq:anbn1b} for $b_n(z)$ we use induction and the fact that from equation \eqref{eq:toda2}, $a_n^2(z)=x_n(z)$ and $b_n(z)=
y_n(z)$ are related by
\[z\deriv{x_n}{t}=x_{n}(y_n-y_{n-1}),\] and initially
\[y_0(z)
= z\deriv{}{z}\big\{\ln \mathcal{T}_1(z;\bb)\big)\}-\frac{\bb}{2}.\]
Hence
\[\begin{split} y_1(z)&=z\deriv{}{z}\big\{\ln x_1(z)\big\}+y_0(z)\\
&= z\deriv{}{z}\left\{\ln \frac{\mathcal{T}_2(z;\bb)\mathcal{T}_0(z;\bb)}{\mathcal{T}_1^2(z;\bb)}\right\}+ z\deriv{}{z}\left\{\ln \mathcal{T}_1(z;\bb)\right\}-\frac{\bb}{2} \\
&= z\deriv{}{z}\left\{\ln \frac{\mathcal{T}_2(z;\bb)}{\mathcal{T}_1(z;\bb)}\right\}-\frac{\bb}{2},
\end{split}\]
since $\mathcal{T}_0(z;\bb)=1$,
so \eqref{deq:anbn1b} is true for $n=1$.
Now suppose that \eqref{deq:anbn1b} is true, then
\[\begin{split} y_{n+1}(z)&=z\deriv{}{z}\big\{\ln x_n(z)\big\}+y_n(z)\\
&=z\deriv{}{z}\left\{\ln\frac{\mathcal{T}_{n+2}(z;\bb)\mathcal{T}_{n}(z;\bb)}{\mathcal{T}_{n+1}^2(z;\bb)}\right\}+z\deriv{}{z}\left\{\ln\frac{\mathcal{T}_{n+1}(z;\bb)}{\mathcal{T}_n(z;\bb)}\right\}-\frac{\bb}2\\
&=z\deriv{}{z}\left\{\ln\frac{\mathcal{T}_{n+2}(z;\bb)}{\mathcal{T}_{n+1}(z;\bb)}\right\}-\frac{\bb}2,
\end{split}\]
as required, and so the result follows by induction. We remark that equation \eqref{eq:toda1} is identically satisfied by $a_n^2(z)$ and $b_n(z)$ given by \eqref{deq:anbn1}.
\end{proof}

\comment{\begin{remark}{\rm The associated Bessel function solution of \eqref{eq:u2n} is given by
\beq
v_n( x)=  x\deriv{}{ x}\left(\ln\frac{\tau_{n+1}( x;\bb)}{\tau_n( x;\bb)}\right)-n+\tfrac12.
\eeq
where
\[
\tau_n( x;\bb)=\det\left[\left( x\deriv{}{ x}\right)^{\!j+k} \BesselI{\bb}(x)\right]_{j,k=0}^{n-1},\]
with $\tau_0( x;\bb)=1$, and $\BesselI{\bb}(x)$ is the modified Bessel function.
}\end{remark}

\begin{lemma}{As $z\to\infty$
\begin{align}
a_n^2(z)= \frac{n\,z^{1/2}}{2}&-\frac{n(4\bb^2-1)}{64\,z^{1/2}}-\frac{n^2(4\bb^2-1)}{64\,z}-\frac{n(4\bb^2-1)(16n^2-4\bb^2+9)}{4096\,z^{3/2}} + \O(z^{-2}),\\
b_n(z)= \frac{n\,z^{1/2}}{2}&+\frac{2n-2\bb+1}{4}+\frac{4\bb^2-1}{16\,z^{1/2}}+\frac{(4n+1)(4\bb^2-1)}{64\,z}\nonumber\\
&+\frac{(4\bb^2-1)(48n^2+24n-4\bb^2+25)}{1024\,z^{3/2}} + \O(z^{-2}).
\end{align}
}\end{lemma}}

\def\dz{\delta_z}
In a recent paper, Fern\'andez-Irisarri and Ma\~nas \cite[\S2]{refFIM} discuss the generalised Charlier weight \eqref{gencharlierw}, in particular properties of the coefficients in the recurrence relation. The relationship between the notations in \cite{refFIM} and those here are $x_n(z)=\ga_n(\eta)$ and $y_n(z)=\b_n(\eta)$ with $z=\eta$. Fern\'andez-Irisarri and Ma\~nas \cite{refFIM} relate $x_n(z)$ and $y_n(z)$ to Okamoto's Hamiltonian for \PIIIp\ \cite{refOkamotoPIII}
and derive two \odes\ for $x_n(z)$.

\begin{enumerate}[(i)]
\item
Equation (45) in \cite[Theorem 4]{refFIM} is the third-order equation
\[\dz\left(\frac{x_n}{z}\left\{\dz^2(\ln x_n)+2x_n\right\}+\frac{n^2z}{x_n}\right)=2x_n,\qquad \dz(f)=z\deriv{f}{z},\]
i.e.
\beq\deriv[3]{x_n}{z}=\frac{1}{zx_n^2}\left(z\deriv{x_n}{z}-x_n\right)\left\{2x_n\deriv[2]{x_n}{z}-\left(\deriv{x_n}{z}\right)^2+n^2\right\}-\frac{4x_n}{z^2}\deriv{x_n}{z}+\frac{2x_n(x_n+z)}{z^3},\label{FIM45}
\eeq and the authors state that this equation ``should have the \p\ property". Equation \eqref{FIM45} can be integrated to give equation \eqref{eq:xn}, with $\bb^2$ as the constant of integration. Since equation \eqref{eq:xn} is equivalent to deg-\PV\ \eqref{eq:degPV0} then equation \eqref{FIM45} does have the \p\ property.
\item
Equation (60) in \cite[Theorem 5]{refFIM} is the second-order equation
\[\begin{split} \left(1-\frac{x_n}{z}\right)&\left\{\dz\left(\frac{\dz(x_n)+nz}{x_n}\right)+2x_n\right\}+2\{x_n-z+(n-b)n\}\\
&=-\tfrac12\left(\frac{\dz(x_n)+nz}{x_n}\right)^2+(n+1)\left(\frac{\dz(x_n)+nz}{x_n}\right)+(n-b-1)(3n-b+1),
\end{split}\]
which is equation \eqref{eq:xn} with
\[\bb^2=2(b-n)^2+n^2-2n-1.\]
\end{enumerate}

\section{\label{sec:dis}Discussion}
In this paper the classical solutions of deg-\PV\ \eqref{eq:degPV0} have been classified.
Ohyama and Okumura \cite[Theorem 2.1]{refOO13} give a list of classical solutions of \PI\ to \PV\ and state that
``deg-P5 with $\a=\tfrac12a^2$, $\b=-\tfrac18$, $\ga=-2$ has the algebraic solution $w(z) = 1+2\sqrt{z}/a$"\footnote{As noted in \cite{refAHvdPT19}, there is typo in \cite{refOO13} who say $\b=-8$ rather than $\b=-\tfrac18$.}
and ``deg-P5 with $\b=0$ has the Riccati type solutions". 
The results in this paper show that there are more classical solutions of deg-\PV\ \eqref{eq:degPV}.
The algebraic solution is equivalent to the ``seed solution" obtained by setting $n=0$ in \eqref{eq:pvalgsol1}, i.e.
\[ u_0(x;\mu)=\frac{\mu}{x+\mu},\]
and there is a more general hierarchy of ``Riccati type solutions" which are described in Theorem \ref{thm:BesselPV0}.

All solutions of \PII--\PVI\ that are expressible in terms of special functions satisfy a first-order equation of the form
\beq \left(\deriv{u}{x}\right)^{\!n}=\sum_{j=0}^{n-1}F_j(u,x) \left(\deriv{u}{x}\right)^{\!j}, \label{eq:1order}\eeq
where $F_j(u,x)$ is polynomial in $u$ with coefficients that are rational functions of $x$. It can be shown that the Bessel function solutions of \PIII\ \eqref{eq:PIII} satisfy a first-order equation of the form \eqref{eq:1order} for $n$ odd, whereas the Bessel function solutions of deg-\PV\ \eqref{eq:degPV0} satisfy a first-order equation of the form \eqref{eq:1order} for $n$ even.

The relationship between \PIII\ \eqref{eq:PIII} and deg-\PV\ \eqref{eq:degPV} is similar to that between the second \peq\ (\PII)
\beq\deriv[2]{q}{x} = 2q^3 + x q,\label{eq:PII}\eeq
with $\a$ a parameter, and  \p\ XXXIV equation ($\mbox{\rm P}_{\!34}$)
\beq\deriv[2]{p}{x} =\frac1{2p}\left(\deriv{p}{x}\right)^{\!2} 
+ 2p^2-xp-\frac{(\a+\tfrac12)^2}{2p},\label{eq:p34}\eeq
which is equivalent to equation XXXIV of Chapter 14 in \cite{refInce}, in that both pairs of equations arise from a Hamiltonian. 
The Hamiltonian associated with \PII\ (\ref{eq:PII}) and $\mbox{\rm P}_{\!34}$ \eqref{eq:p34} is
\begin{equation}\label{sec:PT.HM.DE4}
\HII(q,p,z;\a) = \tfrac12 p^2 - (q^2+\tfrac12 z)p - (\a+\tfrac12)q
\end{equation} and so
\begin{equation}\label{sec:PT.HM.DE3}
\deriv{q}{z}=p-q^2-\tfrac12z,\qquad \deriv{p}{z}=2qp+\a+\tfrac12,
\end{equation} see \cite{refJM81,refOkamoto80ab}.
It is known that \PII\ \eqref{eq:PII} and $\mbox{\rm P}_{\!34}$ \eqref{eq:p34} have special function solutions in terms of Airy functions, cf.~\cite{refPAC16}. It can be shown that the Airy function solutions of \PII\ \eqref{eq:PII} satisfy first-order equation of the form \eqref{eq:1order} for $n$ odd, whereas the Airy function solutions of $\mbox{\rm P}_{\!34}$ \eqref{eq:p34} satisfy a first-order equation of the form \eqref{eq:1order} for $n$ even. Further the function $\sigma(z;\a)=\HII(q,p,z;\a)$ given by \eqref{sec:PT.HM.DE4}, with $q$ and $p$ satisfying \eqref{sec:PT.HM.DE3}, satisfies the second-order, second degree equation (\sPII)
\begin{equation}\label{eq:SII}
\left(\deriv[2]{\sigma}{z}\right)^{\!2}+4\left(\deriv{\sigma}{z}\right)^{\!3}+2\deriv{\sigma}{z}\left(z\deriv{\sigma}{z}-\sigma\right)
=\tfrac14(\a+\tfrac12)^2, \end{equation} 
see \cite{refJM81,refOkamoto80ab}. Conversely, if $\sigma(z;\a)$ is a solution of \eqref{eq:SII}, then 
\begin{equation}\label{sec:PT.HM.DE6}
q(z;\a)=\frac{4\sigma''(z;\a)+2\a+1}{8\sigma'(z;\a)},
\qquad p(z;\a)=-2\sigma'(z;\a),
\end{equation} 
with $'\equiv \d/\d z$, are solutions of \eqref{eq:PII} and \eqref{eq:p34}, respectively. Consequently it is simpler to express classical solutions of $\mbox{\rm P}_{\!34}$ 
\eqref{eq:p34} in terms of classical solutions of \sPII\  \eqref{eq:SII}, which involve one determinant, rather than solutions of \PII\ \eqref{eq:PII}, which involve two determinants.

\section*{Acknowledgements}
I thank Clare Dunning 
and Steffen Krusch for helpful comments and illuminating discussions. I also thank the anonymous reviewers whose comments were invaluable in improving the manuscript.

\def\bitem#1{\vspace{-0.2cm}\bibitem{#1}}

\def\refjl#1#2#3#4#5#6#7{\vspace{-0.2cm}
\bibitem{#1}\textrm{\frenchspacing#2}, \textrm{#3},
\textit{\frenchspacing#4}, \textbf{#5}\ (#6)\ #7.}

\def\refjsm#1#2#3#4#5#6#7{\vspace{-0.2cm}
\bibitem{#1}\textrm{\frenchspacing#2}, \textrm{#3},
\textit{\frenchspacing#4} (#6)\ #7.}

\def\refbk#1#2#3#4#5{\vspace{-0.2cm}
\bibitem{#1} \textrm{\frenchspacing#2}, \textit{#3}, #4, #5.}

\def\refpp#1#2#3#4{\vspace{-0.2cm}
\bibitem{#1}\textrm{\frenchspacing#2}, #3, #4.}

\def\refcf#1#2#3#4#5#6{\vspace{-0.2cm}
\bibitem{#1} \textrm{\frenchspacing#2}, \textrm{#3},
in: \textit{#4},\ {\frenchspacing#5} #6.}

\def\DE{Diff. Eqns.}
\def\JPA{J. Phys. A}
\def\NMJ{Nagoya Math. J.}
\def\CUP{Cambridge University Press}

\end{document}

\begin{theorem}Supppose that $u_0=u(x;\a,\b,\ga)$ satisfies \eqref{eq:degPV0} with parameters 
\[ \a=\tfrac12a^2,\qquad \b=-\tfrac12b^2,\qquad \ga=c.\]
Let $u_j=u(x;\a_j,\b_j,\ga_j)$, $j=1.2.3,4$ be solutions of \eqref{eq:degPV0} with parameters
\[\begin{array}{l@{\qquad}l@{\qquad}l}
\a_1=\tfrac12(a+1)^2,& \b_1=-\tfrac12b^2,& \ga_1=c,\\[2.5pt]
\a_2=\tfrac12(a-1)^2,& \b_2=-\tfrac12b^2,& \ga_2=c,\\[2.5pt]
\a_3=\tfrac12a^2,& \b_3=-\tfrac12(b+1)^2,& \ga_3=c,\\[2.5pt]
\a_4=\tfrac12a^2,& \b_4=-\tfrac12(b-1)^2,& \ga_4=c,
\end{array}\] respectively. Then
\begin{eqnarray*} \fl
\mathcal{U}_1:\quad u_1
&=\frac{\ds\left\{x \ux+2(\u-1) (a\u -b)\right\} \left\{x \ux+2(\u-1) (a\u +b)\right\}}{\ds x^2(\ux)^2+4ax\u(\u-1)\ux+4c\u(\u-1) x^{2}+4(\u-1)^{2} (a^2 \u^2 -b^2)},\\ \fl
\mathcal{U}_2:\quad u_2
&=\frac{\ds\left\{x \ux-2(\u-1) (a\u -b)\right\} \left\{x \ux-2(\u-1) (a\u +b)\right\}}{\ds x^2(\ux)^2-4 ax\u(\u-1)\ux+4c\u(\u-1) x^{2}+4(\u-1)^{2} (a^2 \u^2 -b^2)},\\ \fl
\mathcal{U}_3:\quad u_3
&=\frac{\ds x^2(\ux)^2+4 bx (\u-1)\ux+4 c x^2 \u^{2}(\u-1)-4(\u-1)^{2} (a^2 \u^2 -b^2)}{\ds\left\{x \ux-2(\u-1) (a\u -b)\right\} \left\{x \ux+2(\u-1) (a\u +b)\right\}},\\ \fl
\mathcal{U}_4:\quad u_4
&=\frac{\ds x^2(\ux)^2-4 bx (\u-1)\ux+4 c x^2 \u^{2}(\u-1)-4(\u-1)^{2} (a^2 \u^2 -b^2)}{\ds\left\{x \ux-2(\u-1) (a\u -b)\right\} \left\{x \ux+2(\u-1) (a\u +b)\right\}}.

\comment{\newpage
\[ \begin{split}\tau_n(z;\mu,\ep) &= \det\left[ \left(z\deriv{}{z}\right)^{\!j+k}\ph_{\mu}(z;\ep) \right]_{j,k=0}^{n-1} 
=\left| \begin{matrix} \ph_{\mu} & \delta \ph_{\mu} & \ldots &\delta^{n-1} \ph_{\mu}\\
\delta\ph_{\mu} & \delta^2 \ph_{\mu} & \ldots &\delta^n \ph_{\mu}\\
\vdots & \vdots & \ddots & \vdots \\
\delta^{n-1} \ph_{\mu} & \delta^{n} \ph_{\mu} & \ldots & \delta^{2n-2} \ph_{\mu} \end{matrix} \right|\\
&=\left| \begin{matrix} \ph_{\mu} & \ph_{\mu}^{(1)} & \ldots & \ph_{\mu}^{(n-1)}\\
\ph_{\mu}^{(1)} &\ph_{\mu}^{(2)} & \ldots & \ph_{\mu}^{(n)} \\
\vdots & \vdots & \ddots & \vdots \\
\ph_{\mu}^{(n-1)} & \ph_{\mu}^{(n)} & \ldots & \ph_{\mu}^{(2n-2)}
\end{matrix} \right|,\qquad \ph_{\mu}^{(m)} = \left(z\deriv{}{z}\right)^{\!m}\ph_{\mu}(z;\ep) 
\end{split}\]
with
\[ \ph_{\mu}(z;\ep) =\begin{cases} c_1 J_{\mu}(2\sqrt{z}) + c_2 Y_{\mu}(2\sqrt{z}), &\quad\text{if}\quad \ep=1,\\
c_1 I_{\mu}(2\sqrt{z}) + c_2 I_{-\mu}(2\sqrt{z}), &\quad\text{if}\quad \ep=-1.
\end{cases}\]

\[w_n(z;\mu,\ep)=\frac{\sigma_n'(z;\mu,\ep)+\ep}{\sigma_n'(z;\mu,\ep)-\ep}\]

\[ w_n(z;\mu,\ep)=1+
\ep \frac{z\tau_{n}^2(z;\mu,\ep)}{\tau_{n+1}(z;\mu,\ep)\,\tau_{n-1}(z;\mu,\ep)}\]
\[ w_n(z;\mu,\ep)=\ep z\frac{\tau_{n}(z;\mu+1,\ep)\,\tau_{n}(z;\mu-1,\ep)}{\tau_{n+1}(z;\mu,\ep)\,\tau_{n-1}(z;\mu,\ep)}\]}

\comment{\subsection{\bts\ of \PIII}
Hierarchies of rational solutions of the \p\ equations can be
obtained by applying \bk\ transformations to ``seed solutions". The
\bk\ transformations of \PIII, which relate two solutions of \PIII\ with
different values of the parameters, are defined as follows. Suppose $w=
w(z;\a,\b,1,-1)$ is a solution of
\PIII, then $w_j= w_j(z;\a_j,\b_j,1,-1)$,
$j=1,2,\ldots,6$, are also solutions of \PIII\ where
\begin{align*}
\W_1:\qquad 
w_1&= \frac{z w'+z w^2-\b w-w+z} {w(z w'+z w^2+\a
w+w+z)},&\a_1&=\a+2,&\b_1&=\b+2,\\
\W_2:\qquad 
w_2&= -\,\frac{z w'-z w^2-\b w-w+z}{w(z w'-z w^2 -\a
w+w+z)},&\a_2&=\a-2,&\b_2&=\b+2,\\
\W_3:\qquad 
w_3&=-\,\frac{z w'+z w^2+\b w-w-z} {w(z w'+z w^2 +\a
w+w-z)},&\a_3&=\a+2,&\b_3&=\b-2,\\
\W_4:\qquad 
w_4&=\frac{z w'-z w^2+\b w-w-z}{w(z w'-z w^2-\a
w+w-z)},&\a_4&=\a-2,&\b_4&=\b-2.\\
\W_5:\qquad 
w_5&=-w,&\a_5&=-\a,&\b_5&=-\b\\
\W_6:\qquad 
w_6&=\ifrac{1}{w},&\a_6&=-\b,&\b_6&=-\a,
\end{align*}

\[\begin{array}{|c@{\enskip}|@{\enskip}c@{\enskip}|@{\enskip}c|}
\hline
\W_j & \a_j & \b_j \\ \hline
\W_1 & \a+2 & \b+2\\
\W_2 & \a-2 & \b+2\\
\W_3 
& \a+2 & \b-2\\
\W_4 
& \a-2 & \b-2\\ \hline
\end{array}\qquad \qquad
\begin{array}{|c@{\enskip}|@{\enskip}c@{\enskip}|@{\enskip}c|}
\hline
\W_j \W_k& \a_{j,k} & \b_{j,k} \\ \hline
\W_1\W_2 & \a & \b+4\\
\W_1\W_3 & \a+4 & \b\\
\W_1\W_4 & \a & \b\\
\W_2\W_3 & \a+4 & \b\\
\W_2\W_4 & \a-4 & \b\\
\W_4\W_4 & \a & \b-4\\ \hline
\end{array}\qquad\qquad
\begin{array}{c}\W_4=\W_1^{-1}\\ \W_3=\W_2^{-1} 
\end{array}\]}

\comment{Then
\begin{subequations}\beq\wiii{w}_{m,n}(z;\a)=w\left(z;\wiii{A}_{m,n},\wiii{B}_{m,n},\wiii{C}_{m,n},\wiii{D}_{m,n}\right)
=-\frac{\U_{m,n-1}(z;\a)\U_{m-1,n}(z;\a)}{\U_{m-1,n}(z;\a-2)\U_{m,n-1}(z;\a+2)},\eeq
is a rational solution of \PV\ \eqref{eq:PVgen} for the parameters
\beq\left(\wiii{A}_{m,n},\wiii{B}_{m,n},\wiii{C}_{m,n},\wiii{D}_{m,n}\right)
=\left(\tfrac18\mu^2,-\tfrac18(\mu-2m+2n)^2,-m-n,-\tfrac12\right),\eeq\end{subequations}
and
\begin{subequations}\beq\wv{w}_{m,n}(z;\a)=w\left(z;\wv{A}_{m,n},\wv{B}_{m,n},\wv{C}_{m,n},\wv{D}_{m,n}\right)
=-\frac{\U_{m,n-1}(z;\a+1)\U_{m,n+1}(z;\a-1)}{\U_{m-1,n}(z;\a-1)\U_{m+1,n}(z;\a+1)},\eeq
is a rational solution of \PV\ \eqref{eq:PVgen} for the parameters
\beq\left(\wv{A}_{m,n},\wv{B}_{m,n},\wv{C}_{m,n},\wv{D}_{m,n}\right)
=\left(\tfrac12(m+\tfrac12),-\tfrac12(n+\tfrac12),m-n-\mu,-\tfrac12\right).\eeq\end{subequations}}